\documentclass[final,leqno]{scrarticle}
\usepackage{amsmath}
\usepackage{amssymb}
\usepackage{amsthm}
\usepackage{float}
\usepackage{comment}
\usepackage{graphicx}
\usepackage{textcomp}
\usepackage{appendix}
\usepackage{wrapfig}
\usepackage{subcaption}
\usepackage{srcltx} 
\usepackage{hyperref} 
\usepackage{algorithm}
\usepackage{algpseudocode}
\hypersetup{draft=false, colorlinks=true, linkcolor=black, urlcolor=blue, citecolor=black}
\newtheorem{theorem}{Theorem}[section]
\newtheorem{lemma}[theorem]{Lemma}
\newtheorem{remark}[theorem]{Remark}

\newtheorem{proposition}[theorem]{Proposition}

\newcommand{\norm}[1]{\left\| #1 \right\|}

\newcommand{\RR}{\mathbb{R}}

\newcommand{\calS}{\mathcal{S}}

\newcommand{\lzo}{L^2(\Omega)}

\newcommand{\Htd}{{H^1_\|(\Omega)^d}}
\newcommand{\Hd}{{H^1(\Omega)^d}}
\newcommand{\Hdz}{H^1_0(\Omega)^d}

\newcommand{\lzd}{{L^2(\Omega)^d}}

\newcommand{\dd}{\, \mathrm{d}}

\renewcommand{\div}{\mathrm{div}}

\author{Emily Klass \thanks{Institute of Numerical and Applied Mathematics, University of G\"ottingen, Germany (\texttt{e.klass@math.uni-goettingen.de})}, \and Tram Thi Ngoc Nguyen \thanks{Max Planck Institute for Solar System Research, G\"ottingen, Germany (\texttt{nguyen@mps.mpg.de})},
\and Nilay Cicek \thanks{Institute for Physical Chemistry, University of G\"ottingen, Germany (\texttt{nilay.cicek@uni-goettingen.de})},
\and Yoav G. Pollack,\thanks{Institute for Dynamics of Complex Systems, University of Göttingen, Germany (\texttt{yoavpol@gmail.com})} 
\and Sarah K\"oster \thanks{Institute for X-Ray Physics, University of G\"ottingen, Germany (\texttt{sarah.koester@uni-goettingen.de})}, 
\and Andreas Janshoff \thanks{Institute for Physical Chemistry, University of G\"ottingen, Germany (\texttt{ajansho@gwdg.de})}, 
\and Anne Wald \thanks{Institute of Numerical and Applied Mathematics, University of G\"ottingen, Germany (\texttt{a.wald@math.uni-goettingen.de}),  corresponding author},}

\title{Determination of active forces in actomyosin systems as inverse source problems for the Stokes equation} 
        
\begin{document}

\maketitle

\begin{abstract}
 The identification of forces and stresses is a central task in biophysics research: Knowledge on forces is key to understanding dynamic processes in active biological systems that are able to self-organize and display emergent properties by converting energy into mechanical work. The aim of this paper is to identify forces generated by a filament-motor network of F-actin and myosin -- actomyosin --  and exerted on the surrounding fluid, therefore causing a fluid flow. In particular, we evaluate optical microscopy data stemming from two different physical settings, confined and non-confined active gels. As a theoretical model, we use the Stokes equation together with an incompressibility condition and suitable boundary conditions reflecting the physical settings. The problem of determining the forces from knowledge on the fluid flow is formulated as an inverse source problem. Due to experimental limitations, only incomplete data are available. We provide a rigorous analysis of the forward problems and the impact of missing data, derive the adjoints of the forward operators needed for regularization, and demonstrate our methods on both synthetic and experimentally measured data.  
\end{abstract}

 \vspace{2ex}


 \vspace{2ex}

 \textbf{MSC numbers:}
%
65J05; 65J10, 65J15, 65J22, 65N21,35R25, 35R30, 92C05


\section{Introduction}

Active systems are collections of subunits that continuously convert chemical energy into directed motion or internal stresses, leading to sustained nonequilibrium dynamics and emergent collective behavior. At scales from single molecules to tissues, this persistent energy consumption breaks detailed balance and gives rise to phenomena such as spontaneous flows, pattern formation, and large‑scale organization that cannot be captured by equilibrium statistical mechanics \cite{Ramaswamy2010, Marchetti2013, Bechinger2016, Gompper2020, teVrugt2025}. Understanding how microscopic activity translates into mesoscopic and macroscopic mechanics is therefore a central goal of active matter research, with implications for soft (smart) materials, synthetic biology and cell biology. In principle, precise knowledge of the subunits' positions and momenta (forces) is required, which, however, can be rather intricate to determine considering the complex nature of biological systems and the restrictions imposed by the quality of the experimental data.

A particularly tractable and biologically relevant class of active materials are actomyosin gels \cite{JULICHER2007, Prost2015, Burla2019}. The cytoskeleton of eukaryotic cells contains actin filaments and myosin motors whose ATP‑driven interactions generate contractile stresses and reorganize filament networks. Reconstituted actomyosin systems and dense cellular cortices both display rich mechanical responses: they are viscoelastic, heterogeneous, and capable of producing sustained internal flows and force generation through motor activity \cite{AbuShah2014, Litschel2021, Liebe2023, Sakamoto2024, Sciortino2025, wollrab16}. The biological implications are immense as these processes enable cells to divide, migrate and organize into larger structures -- tissues -- that themselves show emergent out-of-equilibrium behavior but on a different length scale.  

These properties render actomyosin gels a fundamental, yet minimal and controllable platform for investigating how local biochemical activity produces forces and eventually collective motion.
When confined to finite compartments or embedded in a surrounding fluid, actomyosin networks drive internal, non‑equilibrium flows whose spatial structure reflects the distribution and orientation of active stresses \cite{Liebe2023, Litschel2021, Sciortino2025, Tsai2011}. In droplets or other confined geometries, internal circulations, vortical patterns and symmetry‑breaking flows can emerge; in bulk, active fluids, contractile or extensile activity can trigger spontaneous flows and hydrodynamic instabilities \cite{Bashirzadeh2021, Livne2024}. Importantly, the forces and stresses driving these flows are not directly accessible in most experiments, as probe‑based force measurements are often invasive, can perturb the dynamics under study, and are too local to capture the full spatial distribution of forces.

To gain information on these forces, we adopt an inverse approach. In our experiments, we observe steady, laminar flows using optical microscopy, which is followed by particle image velocimetry (PIV) analysis \cite{PIV1,PIV2} to extract the velocity fields. The relation between these data and the underlying forces that must act within the active medium is described with low‑Reynolds‑number hydrodynamic models that account for spatial heterogeneity in viscosity. This formulation casts the problem as a source identification task for Stokes‑type operators, including Stokes–Brinkman variants when porous or network drag is relevant (see also \cite{Lechleiter13}). These turn out to be affine linear inverse problems with incomplete data. To account for measurement noise when solving problems that are typically ill-posed \cite{hadamard1923,bh00}, regularization methods \cite{benning_burger_2018,ehn00,skhk12} are used, relying on our rigorous mathematical analysis of the inverse problems.

We use two experimental model systems, one based on contractile actomyosin networks with free boundaries (bulk) and one where the gel is reconstituted in water-in-oil droplets (confined). The latter serves a minimal cell model, in which emergent phenomena such as self-propulsion are feasible.

From a mathematical point of view, the two experimental settings considered in this article (confined and bulk) require two distinct mathematical models as detailed in Section \ref{Sec:2}. We present the respective function space settings and give a rigorous analysis, including proofs of the well-posedness of the forward problems in Section \ref{Sec:3} and a discussion on the impact of data limitations resulting from the experiment as described in Section \ref{sec:data_aq}. In view of numerical regularization, we derive the adjoints in Section \ref{sec:inv_prob} and present results for both simulated and experimental data in Section \ref{Sec:Numerics} as a proof-of-concept.

\section{The Stokes equation as a forward model} \label{Sec:2}
In this section, we introduce our notation as well as the underlying physics to give a mathematical description of the relation between the forces representing the activity, in short: active forces, exerted by the actomyosin network, and the resulting flow field, which can be measured experimentally.
We do not aim at describing the protein network structure, but instead only the surrounding fluid. Hence the elastic effects on the fluid are neglected and the flow immediately adjusts to changes in the force $\mathbf{f}$, allowing a time-independent model for each time-instance.

For simplicity and to reflect the available data, we consider a two-dimensional setting, neglecting flow outside the image plane (i.e., the focus plane of the microscope). The model is, however, also valid in three dimensions.\\
Let $\Omega \subset \RR^d$, $d=2$, be a bounded Lipschitz domain with its boundary denoted by $\partial\Omega$. 
We assume that the fluid has, in general, a spatially varying viscosity $\eta$ to provide an additional degree of freedom, reflecting possible inhomogeneities. The active force $\mathbf{f}$, the velocity $\mathbf{v}$, and the pressure $p$ are all spatial quantities, depending on the spatial variable $\mathbf{x}$.   
Their relation can be modeled using the Stokes equation
\begin{equation} \label{eq:stokes_basic}
  - \nabla \cdot \big( \eta D(\mathbf{v}) \big) + \nabla p =  \mathbf{f}, \ \  D(\mathbf{v}):= \frac{1}{2} \left( \nabla \mathbf{v} + (\nabla \mathbf{v})^T \right),
\end{equation}
where $p$ is the pressure inside of the droplet and $D(\mathbf{v})$ is the stress tensor. 

\vspace{1ex}

The differential operator $\nabla$ is applied componentwise, i.e.,
\begin{displaymath}
 \nabla \mathbf{v} = \begin{pmatrix}
             \frac{\partial v_1}{\partial x_1} & \cdots & \frac{\partial v_d}{\partial x_1} \\
             \vdots & \ddots & \vdots \\
             \frac{\partial v_1}{\partial x_d} & \cdots & \frac{\partial v_d}{\partial x_d} 
            \end{pmatrix}.
\end{displaymath}
The viscosity $0<\eta \in L^{\infty}(\Omega)$ is a strictly positive function and $\eta D(\mathbf{v})$ is thus a spatial function with values in $\RR^{d \times d}$. Its divergence is the function $\nabla \cdot \big( \eta D(\mathbf{v}) \big) : \RR^{d\times d} \to \RR^d$.  

\vspace{1ex}

We assume that the active fluid is incompressible, which is modeled by additionally postulating that $ \nabla \cdot \mathbf{v} = 0$.
This incompressibility condition results in $ \nabla \cdot (\nabla \mathbf{v})^T $ being equal to $0$ in every component. 
We note that, if $\eta$ is a constant function with $\eta(\mathbf{x}) \equiv \hat{\eta} > 0$ for all $x \in \Omega$, then $\nabla \cdot \big( \eta D(\mathbf{v}) \big) = \hat{\eta} \Delta \mathbf{v}$ in the distributional sense, where $\Delta$ is the Laplace operator and is applied component-wise.

\paragraph{Notation} 
The normal and tangential components of a vector $\mathbf{v}$ are given as $\mathbf{v}_\perp$ and $\mathbf{v}_{\parallel}$ with the outer normal vector being $\mathbf{n}$. The notation of vector spaces is simplified, so that for example $\Hd := H^1(\Omega, \mathbb{R}^d)$ and $\lzd := L^2(\Omega, \mathbb{R}^d)$. Additionally, ``:'' denotes the element-wise or Hadamard product.

The simplest and therefore most commonly used boundary conditions in connection to the Stokes equation are homogeneous Dirichlet boundary conditions and they are discussed often in literature such as \cite{Bof13,Soh01,John10}. Alternatively, the Stokes operator is also sometimes formulated with Neumann boundary conditions \cite{Mit11}. 

In this paper, we introduce Robin boundary conditions and inhomogeneous Dirichlet boundary conditions, that are also discussed in, e.g., \cite{Ray07} and \cite{Amr11}, respectively, and will be the main focus of the upcoming sections.

\subsection{Fluid flow in actomyosin droplets with Robin boundary conditions}

\begin{wrapfigure}{R}{0.5\textwidth}
    \centering
    \def\svgwidth{200pt}
\begingroup%
  \makeatletter%
  \providecommand\color[2][]{%
    \errmessage{(Inkscape) Color is used for the text in Inkscape, but the package 'color.sty' is not loaded}%
    \renewcommand\color[2][]{}%
  }%
  \providecommand\transparent[1]{%
    \errmessage{(Inkscape) Transparency is used (non-zero) for the text in Inkscape, but the package 'transparent.sty' is not loaded}%
    \renewcommand\transparent[1]{}%
  }%
  \providecommand\rotatebox[2]{#2}%
  \newcommand*\fsize{\dimexpr\f@size pt\relax}%
  \newcommand*\lineheight[1]{\fontsize{\fsize}{#1\fsize}\selectfont}%
  \ifx\svgwidth\undefined%
    \setlength{\unitlength}{448.02782038bp}%
    \ifx\svgscale\undefined%
      \relax%
    \else%
      \setlength{\unitlength}{\unitlength * \real{\svgscale}}%
    \fi%
  \else%
    \setlength{\unitlength}{\svgwidth}%
  \fi%
  \global\let\svgwidth\undefined%
  \global\let\svgscale\undefined%
  \makeatother%
  \begin{picture}(1,0.8319879)%
    \lineheight{1}%
    \setlength\tabcolsep{0pt}%
    \put(0,0){\includegraphics[width=\unitlength,page=1]{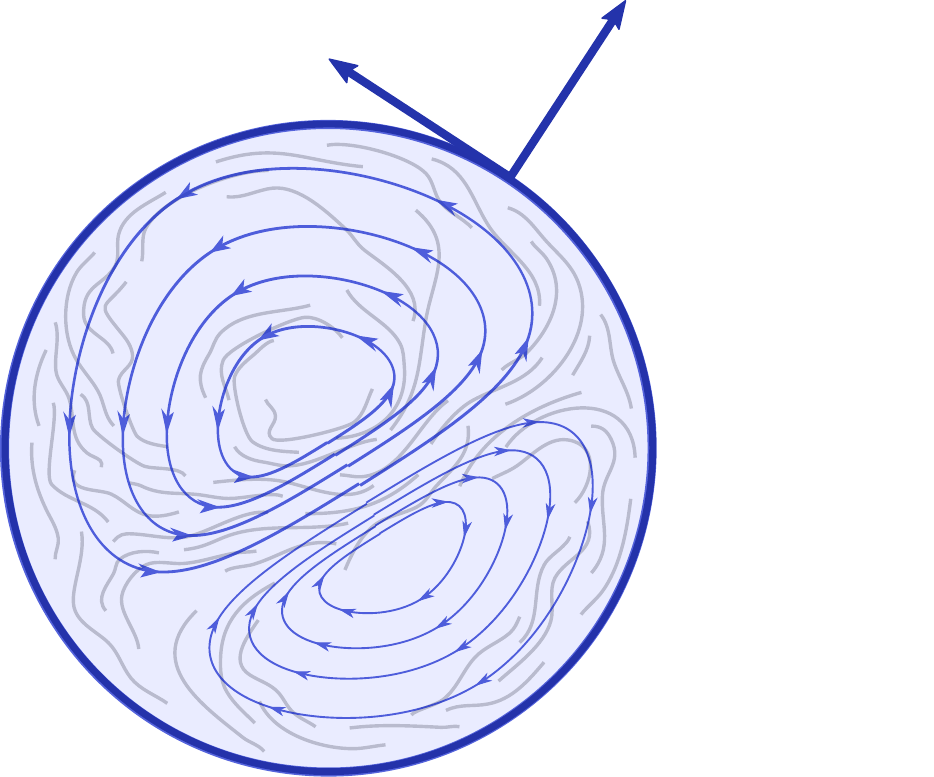}}%
    \put(0.29033461,0.77098222){\color[rgb]{0,0,0}\makebox(0,0)[lt]{\lineheight{1.25}\smash{\begin{tabular}[t]{l}$\mathbf{v}_{\parallel}$\end{tabular}}}}%
    \put(0.67448496,0.78159794){\color[rgb]{0,0,0}\makebox(0,0)[lt]{\lineheight{1.25}\smash{\begin{tabular}[t]{l}$\mathbf{v}_{\bot}$\end{tabular}}}}%
    \put(0.7772908,0.50814399){\color[rgb]{0,0,0}\makebox(0,0)[lt]{\lineheight{1.25}\smash{\begin{tabular}[t]{l}$\partial\Omega$\end{tabular}}}}%
    \put(0.13012509,0.14605201){\color[rgb]{0,0,0}\makebox(0,0)[lt]{\lineheight{1.25}\smash{\begin{tabular}[t]{l}$\Omega$\end{tabular}}}}%
    \put(0,0){\includegraphics[width=\unitlength,page=2]{bc_robin.pdf}}%
  \end{picture}%
\endgroup%

    \caption{Illustration of an actomyosin droplet with domain $\Omega$}
    \label{fig:bc_robin}
\end{wrapfigure}

The droplet is assumed to have a fixed size and shape given by the domain $\Omega$ as a two dimensional disk with boundary $\partial \Omega$. We denote the outward normal of the boundary with $\mathbf{n}$ and  split $\mathbf{v}$ and the co-normal derivative $D(\mathbf{v}) \mathbf{n}$ into a tangential and a normal component, indicated by $\parallel$ and $\bot$ respectively.
Since no fluid leaves or enters the droplet, there is no normal component of the flow field on the boundary, i.e.,
\begin{displaymath}
 \mathbf{v}_{\bot} = \mathbf{0} \quad \text{on } \partial\Omega.
\end{displaymath}
For the tangential part, we postulate Robin-type boundary conditions (slip conditions)
\begin{displaymath} 
 \big(\eta D(\mathbf{v})\mathbf{n} \big)_{\parallel} = -\lambda \mathbf{v}_{\parallel}\quad \text{on } \partial\Omega,
\end{displaymath}

where $\lambda \geq 0$ is the slip length at a rigid boundary and thus a material parameter, which we assume to be known here. The choice of this condition is motivated by the physical setting: According to \cite{KreeZippelius18}, an actomyosin droplet is expected to be able to move or swim through a surrounding fluid, driven by internal active forces. An interaction between the outside fluid and the droplet's fluid is thus expected, but since we do not describe the outside fluid in this simplified model, the use of a slip condition with an appropriate parameter $\lambda$ serves as a suitable replacement. 

The function space for the flow field $\mathbf{v}$ is then given by 
\begin{displaymath}
 H^1_{\parallel}(\Omega)^d := \left\lbrace \mathbf{v} \in H^{1}(\Omega)^{d} \, : \, \mathbf{v}_{\bot} = \mathbf{0} \text{ on } \partial \Omega \right\rbrace .
\end{displaymath}

For the pressure $p$, we use the function space $
 L^2_0(\Omega) :=  \left\lbrace p \in L^2(\Omega) : \int_{\Omega} p \dd x = 0 \right\rbrace
$
to ensure well-definedness of the problem (see Section \ref{Sec:3} for more details).

Altogether we use the Stokes problem
\begin{equation} \label{eq:stokes_problem}
 \begin{split}
  -\nabla \cdot \big( \eta D(\mathbf{v}) \big) + \nabla p &=  \mathbf{f} \quad \text{in } \Omega\\
  \nabla \cdot \mathbf{v} &= 0 \quad \text{in } \Omega\\
  \mathbf{v}_{\bot} &= \mathbf{0} \quad \text{on } \partial\Omega \\
  \big(\eta (\nabla \mathbf{v})\mathbf{n} \big)_{\parallel} &= -\lambda \mathbf{v}_{\parallel}\quad \text{on } \partial\Omega,
 \end{split}
\end{equation}
as a model to describe the relation between the material parameter $\eta$, the pressure $p$, the forces $\mathbf{f}$ and the velocity field $\mathbf{v}$ of the incompressible fluid. 

Before establishing the formulation of our inverse problem, we discuss a different setting motivated by experimental data. 

\subsection{A bulk model for actomyosin-driven flow using Dirichlet boundary conditions\label{sec:bulk}}

\begin{wrapfigure}{R}{0.45\textwidth}
    \centering
    \def\svgwidth{200pt}
\begingroup%
  \makeatletter%
  \providecommand\color[2][]{%
    \errmessage{(Inkscape) Color is used for the text in Inkscape, but the package 'color.sty' is not loaded}%
    \renewcommand\color[2][]{}%
  }%
  \providecommand\transparent[1]{%
    \errmessage{(Inkscape) Transparency is used (non-zero) for the text in Inkscape, but the package 'transparent.sty' is not loaded}%
    \renewcommand\transparent[1]{}%
  }%
  \providecommand\rotatebox[2]{#2}%
  \newcommand*\fsize{\dimexpr\f@size pt\relax}%
  \newcommand*\lineheight[1]{\fontsize{\fsize}{#1\fsize}\selectfont}%
  \ifx\svgwidth\undefined%
    \setlength{\unitlength}{304.49286044bp}%
    \ifx\svgscale\undefined%
      \relax%
    \else%
      \setlength{\unitlength}{\unitlength * \real{\svgscale}}%
    \fi%
  \else%
    \setlength{\unitlength}{\svgwidth}%
  \fi%
  \global\let\svgwidth\undefined%
  \global\let\svgscale\undefined%
  \makeatother%
  \begin{picture}(1,1.0621245)%
    \lineheight{1}%
    \setlength\tabcolsep{0pt}%
    \put(0,0){\includegraphics[width=\unitlength,page=1]{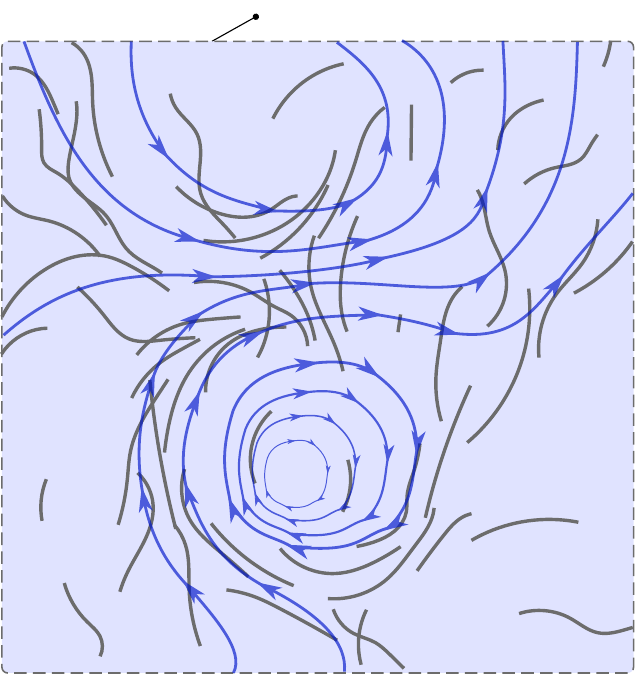}}%
    \put(0.04419055,0.4192164){\color[rgb]{0,0,0}\makebox(0,0)[lt]{\lineheight{1.25}\smash{\begin{tabular}[t]{l}$\Omega$\end{tabular}}}}%
    \put(0.42674817,1.03217309){\color[rgb]{0,0,0}\makebox(0,0)[lt]{\lineheight{1.25}\smash{\begin{tabular}[t]{l}$\partial\Omega$\end{tabular}}}}%
    \put(1.35512976,0.3237264){\color[rgb]{0,0,0}\makebox(0,0)[lt]{\lineheight{1.25}\smash{\begin{tabular}[t]{l}$\partial\Omega$\end{tabular}}}}%
  \end{picture}%
\endgroup%

    \caption{Illustration of an actomyosin network with domain $\Omega$}
    \label{fig:bc_dirichlet}
\end{wrapfigure}
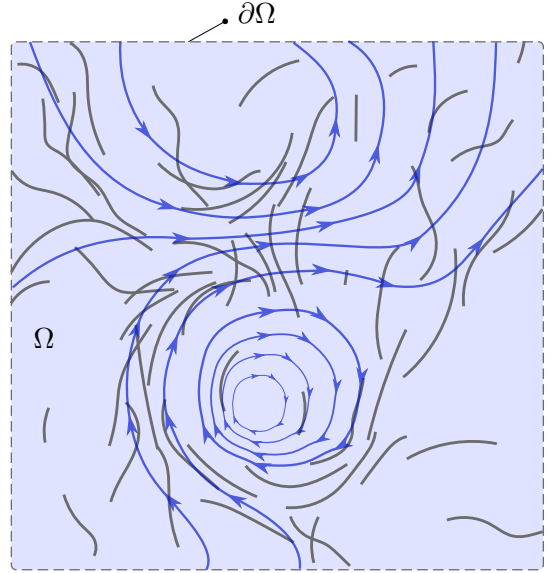

From an experimental point of view it can be useful to measure the velocity field of a network when it is not encapsulated in a droplet or the encapsulation is not visible in the data. Since encapsulating the network into droplets is time-consuming, evaluating the flow of a freely suspended acto-myosin network can be of use to test the network structure and make amendments to the composition of the network and surrounding fluid.

For this version of the problem $\Omega$ is in the shape of a square for $d=2$ or cube for $d=3$. We formulate a Stokes problem with an inhomogeneous Dirichlet boundary condition. That is, we want to find 
$(\mathbf{v},p) \in (\Hd\times L_0^2(\Omega))$ given $\mathbf{f}\in \lzd$ solving
\begin{equation}
\begin{split}
- \nabla \cdot (\eta D(\mathbf{v})) + \nabla p &= \mathbf{f}  \text{ in} \ \Omega\\
\nabla \cdot  \mathbf{v} &= 0  \text{ in} \ \Omega  \\
\mathbf{v} &= \mathbf{g}  \text{ on }  \partial \Omega , 
 \end{split} 
 \end{equation}
for given viscosity $0<\eta\in L^\infty(\Omega)$ and $\mathbf{g}\in H^{\frac{1}{2}}(\partial\Omega)$.

While we expect a rotational flow pattern in the droplets, here it is likely that the rotation-free component is a more prominent part of the overall force. The decomposition into rotation- and divergence-free components of the force will be discussed in section \ref{sec:helmholtz}.

\section{Well-posedness of the forward problems\label{Sec:3}} 
With the proposed model equations, we now examine unique existence of a velocity and pressure field $(\mathbf{v},p)$ given any force $\mathbf{f}$. This guarantees well-definedness for the forward map for the inverse problems in Section \ref{sec:inv_prob}.
We shall discuss both versions of the Stokes models.

\subsection{Well-posedness for the Stokes problem with Robin boundary conditions}

Our analysis is based on weak formulation 
via the bilinear forms
\begin{displaymath}
    a:  H^1_{\parallel}(\Omega)^d \times H^1_{\parallel}(\Omega)^d  \to \RR, \quad
    a(\mathbf{u},\mathbf{v}) :=  \int_{\Omega} \eta D(\mathbf{u}) : D(\mathbf{v}) \dd \mathbf{x} +\int_{\partial\Omega} \mathbf{u}_{\parallel}^T \lambda \mathbf{v}_{\parallel} \dd \boldsymbol{\sigma}  ,
\end{displaymath}
and
\begin{displaymath}
 b: L^2_0(\Omega) \times H^1_{\parallel}(\Omega)^d  \to \RR, \quad \ b(p,\mathbf{u}) := - \int_{\Omega} p (\nabla \cdot\mathbf{u}) \dd \mathbf{x},
\end{displaymath}
as well as the linear form
\begin{displaymath}
 F: H^1_{\parallel}(\Omega)^d \to \RR, \quad \ F(\mathbf{u}) :=  \int_{\Omega} \mathbf{u}^T \mathbf{f} \dd \mathbf{x}.
\end{displaymath}

In our setting, we assume $\eta$ to be known and consider the force-to-solution or Stokes operator 
\begin{equation} \label{eq:S_R}
 \calS_R : L^2(\Omega)^d \to \left( H^1_{\parallel}(\Omega)^d \times L_0^2(\Omega) \right) ,\quad \mathbf{f} \mapsto (\mathbf{v},p),
\end{equation}
that maps the force $\mathbf{f} \in L^2(\Omega)^d$ to the pair of velocity and pressure fields $(\mathbf{v},p) \in H^1_{\parallel}(\Omega)^d \times L_0^2(\Omega)$, which solves the variational formulation 
\begin{equation} \label{eq:stokes_problem_variational}
 \begin{split}
  a(\mathbf{v},\mathbf{u}) + b(p,\mathbf{u}) &= F(\mathbf{u}) \quad \text{for all } \mathbf{u} \in H^1_{\parallel}(\Omega)^d, \\
  b(q,\mathbf{v}) &= 0 \quad \text{for all } q \in L_0^2(\Omega).
 \end{split}
\end{equation}
The subscript $R$ of the operator refers to the Robin boundary condition.

We now provide a well-definedness result for ${\calS}_R$. While our proof is based on the celebrated inf-sup condition as in \cite{Bof13, Hackbusch17, Lechleiter13}, it is necessary to adapt the strategy to the specific boundary condition we are postulating here.

\begin{theorem}[Well-posedness of $\calS_R$] \label{lemma:exuni_Stilde}
The Stokes operator $\calS_R: L^2(\Omega)^d \ni\mathbf{f} \mapsto (\mathbf{v},p)\in H^1_{\parallel}(\Omega)^d \times L_0^2(\Omega)$ satisfying the variational form
 \begin{align}  \int_{\Omega} \eta D(\mathbf{u}) : D(\mathbf{v}) \dd \mathbf{x}
    + \int_{\partial\Omega}\lambda  \mathbf{u}_{\parallel}^T \mathbf{v}_{\parallel} \dd \boldsymbol{\sigma} 
    - \int_{\Omega} p \nabla( \cdot\mathbf{u}) \dd \mathbf{x} &+ \int_{\Omega} q (\nabla \cdot\mathbf{v}) \dd \mathbf{x} = \int_{\Omega} \mathbf{u}^T \mathbf{f} \dd \mathbf{x}\\& 
    \forall \mathbf{u} \in H^1_{\parallel}(\Omega)^d, \forall v\in L^2_0(\Omega) \nonumber
 \end{align}
 is well-defined. Moreover, we have the stability result
 \begin{equation}
     \label{eq:stability}
  \|\mathbf{v}\|_\Htd+\|p\|_{\lzo} \leq C\|\mathbf{f}\|_\lzd   
 \end{equation}
 with some positive constant $C$.
\end{theorem}

\begin{proof}

We define the solution space as $H^1_{\parallel}(\Omega)^d\times  L_0^2(\Omega)$, where $\| \mathbf{v} \|_{\Htd}^2: = \| D(\mathbf{v}) \|_{{\lzo}^{d\times d}}^2 + \norm{\mathbf{v}}^2_{L^2(\partial\Omega)^d}$ is an equivalent norm to the standard $H^1$-norm. Indeed, there exists a constant $c>0$ such that $\norm{\mathbf{v}}_{\Hd}\leq c \norm{\mathbf{v}}_{\Htd}\leq \norm{\mathbf{v}}_{\Hd}$, where the lower bound follows from Korn's inequality  \cite[Lemma A.2(c)]{Lechleiter13}, and the upper bound is an application of the trace theorem on Sobolev spaces.

With this equivalent norm and the positivity of $\eta$ and $\lambda$, we obtain coercivity of the bilinear form $a$ as 
\begin{equation}
\begin{split}
\label{eq:coercivity}
  \left\lvert a(\mathbf{u},\mathbf{u}) \right\rvert 
  &= \left\lvert  \left( \int_{\Omega} \eta D (\mathbf{u}) : D (\mathbf{u}) \dd \mathbf{x} +  \int_{\partial \Omega}\lambda \mathbf{u}_\parallel^T \mathbf{u}_\parallel \dd \boldsymbol{\sigma}   \right)  \right\rvert 
    \geq \min\lbrace\eta,\lambda\rbrace \left\lVert \mathbf{u} \right\rVert_{H^1_{\parallel}(\Omega)^d}^2=:\alpha \left\lVert \mathbf{u} \right\rVert_{H^1_{\parallel}(\Omega)^d}^2.
\end{split}
\end{equation}

For any $q\in L_0^2(\Omega)$, there exists $\mathbf{u}\in \Htd$ such that  $q = \nabla \cdot \mathbf{u}$ and $\left\lVert \mathbf{u} \right\rVert_{H_\parallel^1(\Omega)^d} \leq \frac{1}{\beta} \left\lVert q \right\rVert_{L_0^2(\Omega)}$, thus
\begin{equation}
\label{eq:inf_sup}
\beta \leq \inf_{\norm{q}_{L_0^2(\Omega)}=1} \ \sup_{\norm{\mathbf{u}}_{\Htd} = 1} \frac{\norm{q}_{L_0^2(\Omega)}}{\norm{\mathbf{u}}_{\Htd}}
   = \inf_{\norm{q}_{L_0^2(\Omega)}=1} \sup_{\norm{\mathbf{u}}_{\Htd} = 1}b(\mathbf{u},q)
   \end{equation}
with constant $\beta > 0$ holding for all $q\in L_0^2(\Omega)$. Boundedness of $a$ and $b$ is also clear by similar estimation; see  \eqref{eq:p_stability}.
Boundedness together with the inf-sup condition enables us to claim unique existence of $(\mathbf{v},p)\in H^1_{\parallel}(\Omega)^d\times  L_0^2(\Omega)$ .

As $\mathbf{v}$ is the solution to the Stokes equation it is incompressible, leading to to $b(p,\mathbf{v}) = 0$, hence
\begin{equation}
	\label{eq:v_stability}
    \alpha \| \mathbf{v} \|^2_\Htd  \leq | a(\mathbf{v},\mathbf{v}) |= | F(\mathbf{v}) | \leq \| \mathbf{f} \|_\lzd \| \mathbf{v} \|_\Htd.
\end{equation}

Similarly, we estimate the pressure $p$ using \eqref{eq:stokes_problem_variational}, \eqref{eq:inf_sup} and \eqref{eq:v_stability} as
\begin{equation}
\label{eq:p_stability}
\begin{split}
   \| p \|_\lzd & \leq \frac{1}{\beta} \sup_{\| \mathbf{u} \|_{\Htd}= 1} | b (\mathbf{u},p) |
   = \frac{1}{\beta} \sup_{ \| \mathbf{u} \|_V= 1}  | F (\mathbf{u}) - a(\mathbf{v},\mathbf{u}) | \\
   & \leq \frac{1}{\beta} \sup_{ \| \mathbf{u} \|_{\Htd}= 1} \left( \|\mathbf{u}\|_{\Htd} \|\mathbf{f}\|_\lzd + \max \lbrace \eta, \lambda \rbrace \norm{\mathbf{u}}_{\Htd} \norm{\mathbf{v}}_{\Htd} \right)  \\
   &
   \leq \frac{1}{\beta} \left( \|\mathbf{f}\|_{\lzd} + \frac{\max\lbrace \eta,\lambda\rbrace }{\alpha} \| \mathbf{f} \|_\lzd \right) 
   \end{split}
\end{equation}

The stability result  \eqref{eq:stability} is then achieved with constant $C:= \frac{1}{\alpha}+\frac{1}{\beta}(1+\frac{\max\lbrace \lambda,\eta\rbrace}{\alpha}) >0$ by adding the estimates \eqref{eq:v_stability} and \eqref{eq:p_stability}.
\end{proof}

We note that well-posedness for the Stokes problems for incompressible as well as compressible fluids can also be obtained via the T-coercivity concept introduced by \cite{Cia25}; see also \cite{Hal21}. 

\subsection{Well-posedness for the Stokes problem with Dirichlet boundary conditions}

We now address the Stokes model with inhomogeneous Dirichlet boundary. We begin by introducing a vector function $\mathbf{v}_g \in H^2(\Omega)^d$, such that $\mathbf{v}_g=\mathbf{g}$ on $\partial \Omega$ and that $\mathbf{v}_g$ is incompressible.
Due to the divergence theorem, such $\mathbf{v}_g$ exists if the boundary term $g$ satisfies the \emph{compatibility condition}
\begin{equation}\label{eq:compatibility}
    \int_{\partial \Omega} \mathbf{g}^T \mathbf{n} \dd \boldsymbol{\sigma} = 0 .
\end{equation}
Setting $\mathbf{v}_0:=\mathbf{v}-\mathbf{v}_g$ leads to $\mathbf{v}_0$  satisfying the Stokes equation with homogeneous  boundary
\begin{equation}\label{eq:stokesD_variational_dir}
\begin{split} 
- \nabla \cdot \eta D(\mathbf{v}_0) + \nabla p & 
= \mathbf{f} + \nabla \cdot \eta D(\mathbf{v}_g)\\
 \nabla \cdot  \mathbf{v}_0 &= 0  \text{ in} \ \Omega  \\
\mathbf{v}_0 &=0  \text{ on }  \partial \Omega. 
   \end{split}
\end{equation}
This allows us to define the space for the velocity field $\mathbf{v}_0$ as \\$\Hdz := \lbrace \mathbf{u}\in\Hd : \mathbf{u} = \mathbf{0} \text{ on } \partial \Omega \rbrace$, then the bilinear forms 
\begin{displaymath}
        a:  \Hdz \times \Hdz \to \RR, \quad
    a(\mathbf{u},\mathbf{v_0}) :=  \int_{\Omega} \eta D(\mathbf{u}) : D(\mathbf{v_0}) \dd \mathbf{x}  
\end{displaymath}
and 
\begin{displaymath}
     b:  L^2_0(\Omega) \times \Hdz  \to \RR, \quad \ b(p,\mathbf{u}) := - \int_{\Omega} p (\nabla \cdot\mathbf{u}) \dd \mathbf{x}
\end{displaymath}
as well as the linear form 
\begin{displaymath}
    F: \Hdz \rightarrow \RR,\quad   F(\mathbf{u}) = (\mathbf{f} +\nabla \cdot \eta D(\mathbf{v}_g), \mathbf{u}),
\end{displaymath}
forming weak formulation %
\begin{displaymath} \label{eq:stokes_problem_variational_dir}
 \begin{split}
  a(\mathbf{v}_0,\mathbf{u}) + b(p,\mathbf{u}) &= F(\mathbf{u}) \quad \text{for all } \mathbf{u} \in \Hdz, \\
  b(q,\mathbf{v}_0) &= 0 \quad \text{for all } q \in L^2_0(\Omega)
 \end{split}
\end{displaymath}
We define the source-to-state operator 
as
\begin{displaymath}
    \label{eq:S_D}
     \calS_D : L^2(\Omega)^d \to \left( \Hd \times L_0^2(\Omega) \right),\quad \mathbf{f} \mapsto (\mathbf{v},p) ,
\end{displaymath}
 where $(\mathbf{v}_0,p)$ is the  solution to the variational problem \eqref{eq:stokesD_variational_dir}.

\begin{theorem}[Well-posedness of $\mathcal{S}_D$]\label{th:S_D_well_posed}
The Stokes operator $\calS_D: L^2(\Omega)^d \ni\mathbf{f} \mapsto (\mathbf{v},p)\in H^1(\Omega)^d \times L_0^2(\Omega)$ 
 is well-defined. Moreover, we have the stability result
  \begin{equation}
     \label{eq:stabilityD}
  \|\mathbf{v}\|_\Hd+\|p\|_{\lzo} \leq C\|\mathbf{f}\|_\lzd.   
 \end{equation}
\end{theorem}
\begin{proof}
The proof is analogous to that of Theorem \ref{lemma:exuni_Stilde}, adopting  the solution space $\Hdz \times L_0(\Omega)^d$ for $(\mathbf{v}_0,p)$ and the shifting $\mathbf{v}:=\mathbf{v}_0+\mathbf{v}_g$.
\end{proof}

\section{Data acquisition and analysis\label{sec:data_aq}}

In the following section we describe the workflow how we obtain the data from our experimental setup.

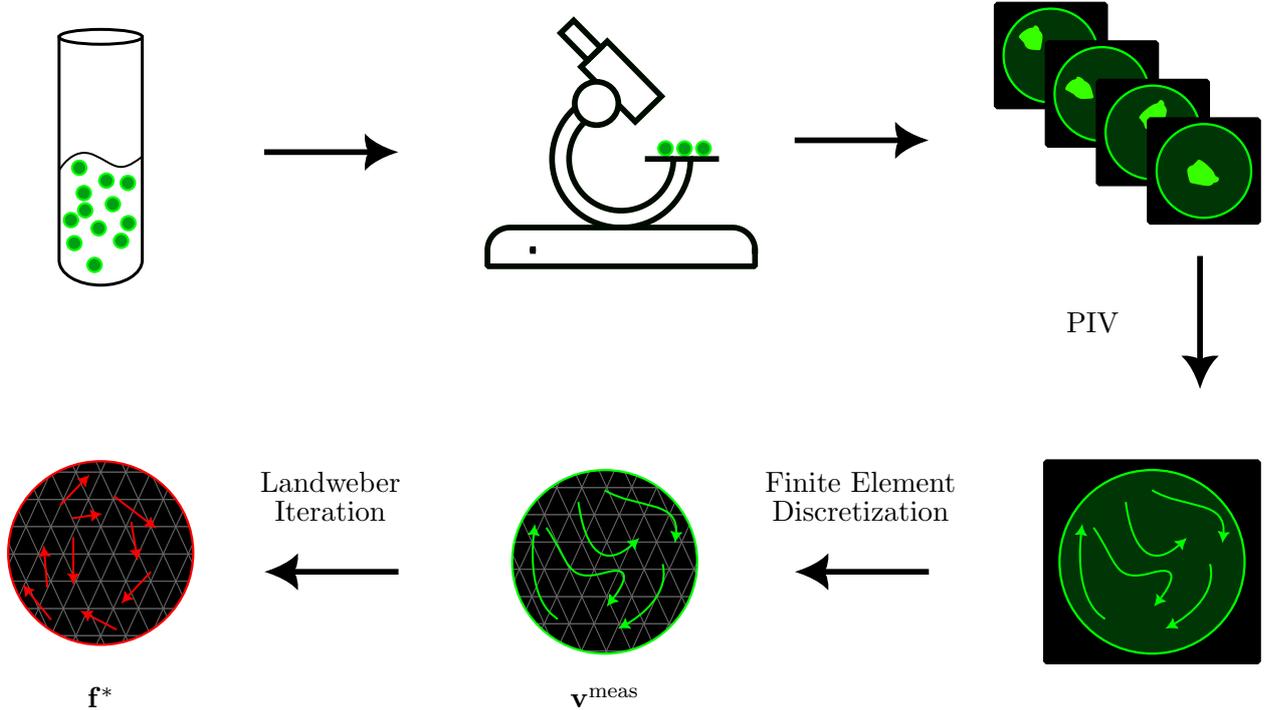
\begin{figure}[H]
    \centering
    \def\svgwidth{\linewidth}
\begingroup%
  \makeatletter%
  \providecommand\color[2][]{%
    \errmessage{(Inkscape) Color is used for the text in Inkscape, but the package 'color.sty' is not loaded}%
    \renewcommand\color[2][]{}%
  }%
  \providecommand\transparent[1]{%
    \errmessage{(Inkscape) Transparency is used (non-zero) for the text in Inkscape, but the package 'transparent.sty' is not loaded}%
    \renewcommand\transparent[1]{}%
  }%
  \providecommand\rotatebox[2]{#2}%
  \newcommand*\fsize{\dimexpr\f@size pt\relax}%
  \newcommand*\lineheight[1]{\fontsize{\fsize}{#1\fsize}\selectfont}%
  \ifx\svgwidth\undefined%
    \setlength{\unitlength}{474.79008712bp}%
    \ifx\svgscale\undefined%
      \relax%
    \else%
      \setlength{\unitlength}{\unitlength * \real{\svgscale}}%
    \fi%
  \else%
    \setlength{\unitlength}{\svgwidth}%
  \fi%
  \global\let\svgwidth\undefined%
  \global\let\svgscale\undefined%
  \makeatother%
  \begin{picture}(1,0.57199113)%
    \lineheight{1}%
    \setlength\tabcolsep{0pt}%
    \put(0,0){\includegraphics[width=\unitlength,page=1]{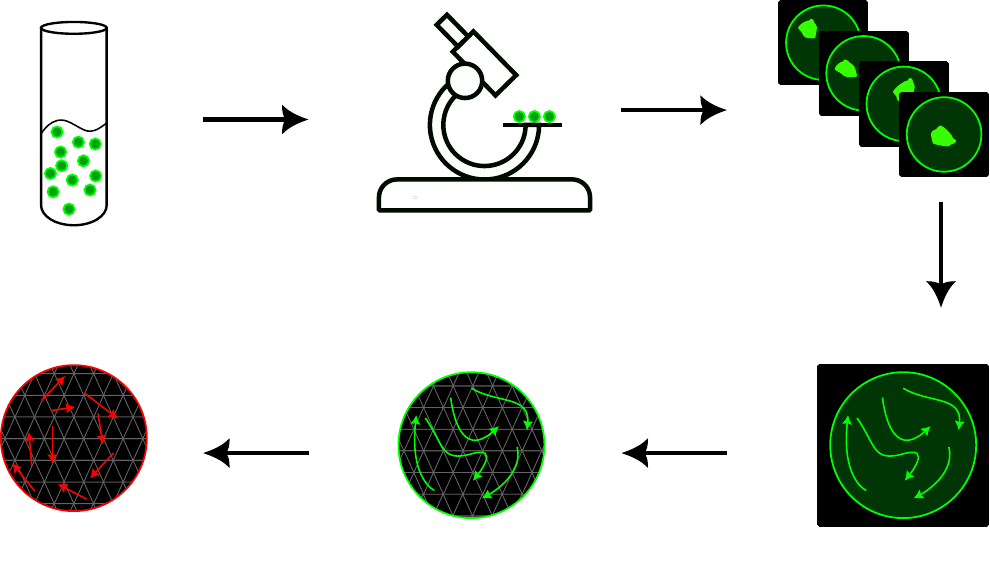}}%
    \put(0.47668413,0.00480783){\makebox(0,0)[t]{\smash{\begin{tabular}[t]{c}$\mathbf{v}^{\operatorname{meas}}$\end{tabular}}}}%
    \put(0.07474659,0.00480783){\makebox(0,0)[t]{\smash{\begin{tabular}[t]{c}$\mathbf{f}^*$\end{tabular}}}}%
    \put(0.84477648,0.30667179){\makebox(0,0)[lt]{\smash{\begin{tabular}[t]{l}PIV\end{tabular}}}}%
    \put(0.68073534,0.17854818){\makebox(0,0)[t]{\smash{\begin{tabular}[t]{c}Finite Element\\Discretization\end{tabular}}}}%
    \put(0.25787317,0.17854818){\makebox(0,0)[t]{\smash{\begin{tabular}[t]{c}Landweber\\Iteration\end{tabular}}}}%
  \end{picture}%
\endgroup%

    \caption{Experiment-to-reconstruction pipeline for actomyosin droplets}
    \label{fig:schematic}
\end{figure}

\begin{wrapfigure}{R}{20em}
        \includegraphics[width=\linewidth]{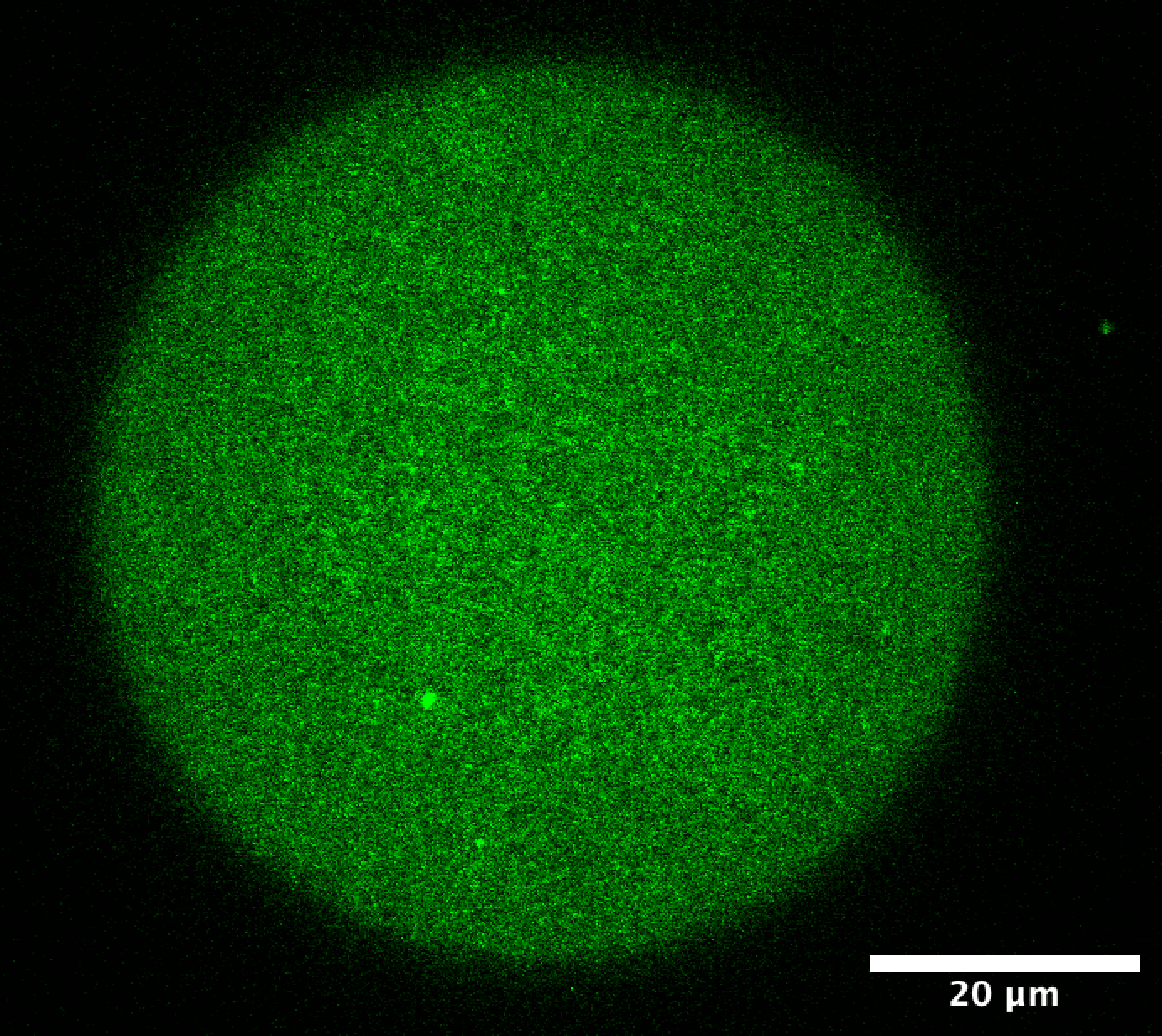}
    \caption{
    Fluorescence image of actin-myosin networks encapsulated in water-in-oil droplets. Actin is fluorescently labelled.} 
\end{wrapfigure}

We prepared actomyosin networks following the protocol described previously \cite{Nietmann2023}. Fluorescently labeled actin filaments and myosin motors were combined to form an entangled, contractile, three‑dimensional network in aqueous buffer solution containing ATP as the chemical energy source. The samples were imaged by time-lapse fluorescence microscopy and the resulting image sequences were analyzed using particle image velocimetry (PIV) to extract the two‑dimensional velocity field of the network (input data). The measured velocity field therefore reports the mesoscale flow of the actomyosin network driven by myosin motor activity in the ATP‑containing buffer, rather than the motion of the solvent alone.
Besides these bulk samples we also encapsulates the contractile networks into water-in-oil droplets to create PIV-data corresponding to the Stokes operator with Robin boundary conditions $\mathcal{S_R}$.

The sample was imaged at 5‑minute intervals beginning immediately after preparation using a confocal laser‑scanning microscope fitted with a 60× water‑immersion objective. Time‑lapse videos were recorded with the Galvano scanner operating in one‑way mode. Comparable experimental setups have been reported previously \cite{Zhang2023,Shimane2022,Pautot2003}, and a detailed description of materials and methods is provided in the Supplementary Information.

Particle Image Velocimetry (PIV) as implemented in the python package \emph{OpenPIV} \cite{alex_liberzon_2020_3930343} was used to extract the 2-dimensional vector flow field.
For preliminary work, the implementation in the software \emph{PIVlab} \cite{Thielicke_2021} was used. PIV deduces the flow field by finding pixel correlations between subsequent frames. This method has previously been used, e.g., for tracking similar active networks in \cite{Sanchez12}. The measured velocity field is used to convert to physical meaningful units depending on the dimensions of the measured area and frame rate.
Since the current study is concerned with a stable steady-state flow, and since the flow in the experiment indeed does not noticeably change throughout the duration of imaging, the entire image stack was considered at once (as an ensemble) producing a single flow field measurement, rather than a time series of flow fields for each frame pair. 
While the actin bundle flow provides a reasonable estimate of the velocity field of the underlying fluid, it does not necessarily adhere to the same fluid dynamics constraints. Particularly, in contrast to the flow of the liquid solution, the flow of the filament network does not need to be incompressible.
However, since the Landweber iteration minimizes the difference between the measured velocity field and the velocity field from the reconstructed force, the method finds a force that produces a flow field of the fluid that is closest of the measured flow field of the actin filaments while still complying to the fluid dynamic constraints such as the incompressibility condition. 

The procedure to obtain the reconstructed force $\mathbf{f}^*$ from the sample containing droplets is sketched in Figure \ref{fig:schematic}.

\subsection{Helmholtz decomposition: divergence-free and curl-free force components \label{sec:helmholtz}}

While data of the velocity flow field can be obtained in the experimental setting described in section \ref{Sec:ExpData}, we do not have data of the pressure field $p$.
This section therefore discusses, which components of the force $\mathbf{f}$ can be reconstructed using only measured data of the velocity field $\mathbf{v}$ for Robin boundary conditions corresponding to the forward operator $S_R$.
We construct two different decompositions of the force $f$: the Helmholtz decomposition into a divergence- and rotation-free component and a decomposition into components dependent on the velocity and pressure terms. We will then see under which assumptions these decompositions coincide.

Let us now assume that the viscosity $\eta$ is constant in $\Omega$ and define the vector spaces
\begin{displaymath}
    V_{\div} (\Omega)^d := \lbrace \mathbf{w} \in L^2(\Omega)^d :  \operatorname{div}(\mathbf{w}) = 0 , \ \mathbf{w} \cdot n = 0 \text{ on } \partial \Omega \text{ in the sense of traces} \rbrace
\end{displaymath}
and
\begin{displaymath}
V_{\operatorname{curl}}(\Omega)^d := \lbrace \mathbf{w} \in \lzd \ : \ \exists\, r \in L^2(\Omega)^{d\times d} \ \text{with} \ \mathbf{w} = - \nabla r \rbrace
\end{displaymath}
with $\left( V_{\div}(\Omega)^d \right)^\perp = V_{\operatorname{curl}}(\Omega)^d $ in $\lzd$.
There exists an $r \in \lzd$, so that the decomposition via the Helmholtz operator $P_{\mathrm{helm}} : \lzd \rightarrow  V_{\div}(\Omega)^d$ is given by  
\begin{displaymath}
    \mathbf{f} = P_{\operatorname{div}} \mathbf{f} - \nabla r,
\end{displaymath}
where $P_{\operatorname{div}} \mathbf{f} \in V_{\div}(\Omega)^d $ is the divergence-free (or solenoidal) and $- \nabla r \in V_{\operatorname{curl}}(\Omega)^d$ rotation-free (or curl-free) component of the force.
Using the coercivity of $a(\mathbf{v},\mathbf{v})$ shown in (\ref{eq:coercivity}) and the incompressibility condition of $\mathbf{v}$ we have
\begin{displaymath} 
    \begin{split}
    \alpha \| \mathbf{v} \|^2_\Htd &\leq | a (\mathbf{v},\mathbf{v}) | = | F(\mathbf{v}) |     \leq \|  P_{\operatorname{div}} \mathbf{f} \|_{\lzd} \| \mathbf{v} \|_{\Htd} ,
    \end{split}
\end{displaymath}
where $(\nabla r,\mathbf{v}) = 0$ since $ \mathbf{v} \in \left( \Htd \cap  V_{\div}(\Omega)^d \right) $.
This provides the continuous dependence of $\mathbf{v}$ on the divergence-free component of the force.

We now want to check if the divergence-free component of the force is equal to the term containing $\mathbf{v}$ and the irrotational component is equal to the term containing $p$.
The Stokes equation is decomposed into terms, one depending on the velocity and one depending on the pressure, so that $\mathbf{f}=\mathbf{f}^{({\mathbf{v}})}+\mathbf{f}^{({p})}$, where
\begin{displaymath}
  \mathbf{f}^{(\mathbf{v})} := - \nabla \cdot \big( \eta D(\mathbf{v}) \big) \ \ \text{and} \ \ \mathbf{f}^{(p)} := \nabla p .
\end{displaymath}
Since the rotation of a gradient-field is always zero, it follows directly that 
$ \nabla \times \mathbf{f}^{(p)} = \nabla \times \left(\nabla p \right) = 0$.
In contrast, the divergence of the velocity term can be rewritten, that is
\begin{equation} \label{eq:divv}
\nabla \cdot \mathbf{f}^{(\mathbf{v})} =  - \nabla \cdot \left( \nabla \cdot \left( \eta D(\mathbf{v}) \right)\right) =  - \left( 2 \nabla \eta \cdot \Delta \mathbf{v} + D(\mathbf{v}) : \nabla(\nabla\eta) \right).
\end{equation}
A detailed calculation is provided in Lemma \ref{lemma:div_appendix}.

Equation \eqref{eq:divv} is zero, if $\nabla\eta=0$, e.g. if the viscosity $\eta$ is constant. 
Then, $\nabla \cdot \mathbf{f}^{(\mathbf{v})}$ is zero and the separation of the velocity and pressure dependent terms of the Stokes equation coincides with the Helmholtz decomposition, so that $\mathbf{f}^{(\mathbf{v})}$ and $\mathbf{f}^{(p)}$ are the divergence-free and irrotational components respectively. If considering solely irrotational components, one can use a stream function representation resulting in a fourth-order equation for $\mathbf{v}$, e.g.~\cite{nguyeninertial}. We do not proceed in this direction, but instead will establish an inversion framework based on the original Stokes equations with incomplete data. 

In the model of the actin network in the experiment's droplet, the gradient of the viscosity corresponds to changes in filament density. In biological cells the density of actin filaments is highly dependent on the cell type and function. However, a typical actin concentration in the experimentally observed droplets is only $12\mu M$ and the volume it takes up is negligible compared to the surrounding fluid. 
We can therefore assume for the experimental data considered here that $\nabla \eta \approx 0$.

If the pressure $p$ is given, it suffices to formulate a reconstruction method for $\mathbf{f}^{(\mathbf{v})}$ from measured data of the velocity field $\mathbf{v}^{\operatorname{meas}}$. 
It is then possible to give a full reconstruction of $\mathbf{f}$ using the gradient of $p$ to recover $\mathbf{f}^{(p)}$. 
Especially in the case of actomyosin droplets, i.e., when a confined fluid is considered, it is valid to assume that the divergence-free component is dominant and allows the determination of the most prominent part of the force $\mathbf{f}$.

\begin{remark}[Source inversion with incomplete data]
    With the forward problems established, we now address the inverse source problem with incomplete data,     referring to the fact that we try to reconstruct the force using only measured velocity and have no information about the pressure. In general, an incomplete data scenario can also arise in the context of statistical data \cite{Evans2002}, imaging with partial measurement \cite{Moscoso2020}, limited observation in inverse scattering \cite{GaoZhang2021}, radial source problems \cite{Hubenthal2011}
    as well as in limited angle tomography \cite{Bubba_2019, Frikel_2013}. 
\end{remark}

\section{Active force identification as an inverse source problem with incomplete data \label{sec:inv_prob}}

We here consider both versions of the Stokes problem in parallel and define a general Stokes operator 
\begin{displaymath}
    \calS \in \lbrace \calS_R, \calS_D \rbrace .
\end{displaymath}
Due to unavoidable measurement noise in experimental data, it is not a realistic assumption on the measured velocity field $\mathbf{v}^{\text{meas}}$ to lie in $\Hd$ and we rather assume that it is set in $\lzd$. Additionally, we assume, that the force $\mathbf{f}$ is in $\Hd$ to allow reconstructions on the boundary. 

We define the embedding operator 
\begin{displaymath}
    E : \left(\Hd \times \lzo\right) \to \left(\lzd \times \lzo \right)
\end{displaymath}
and an operator as a composition of the Stokes operator and the embedding, that is 
\begin{equation}\label{eq:S_embedded}
    S : \Hd  \rightarrow \left(\lzd \times \lzo\right) , \quad\mathbf{f} \mapsto (E \circ \calS ) (\mathbf{f}) ,
\end{equation}
with 
\begin{displaymath}
    S_R:= E \circ \calS_R \quad \text{and} \quad S_D:=E\circ \calS_D,
\end{displaymath}

where $\mathcal{S}$ is well-defined as we have shown in section \ref{Sec:3}.
We will refer to the operator $S$ as the \emph{embedded Stokes operator}.

We define a linear bounded observation operator that only measures the velocity field
\begin{displaymath}
    M : \left( \lzd \times \lzo \right) \rightarrow \lzd , \quad(\mathbf{v}, p ) \mapsto \mathbf{v}
\end{displaymath}
and define the forward operator for the inverse problem
\begin{equation}
    A : \Hd \rightarrow \lzd , \quad\mathbf{f} \mapsto \left( M \circ S \right)(\mathbf{f})  ,
    \label{eq:A_def}
\end{equation}
with
\begin{equation}\label{A}
    A_R := M\circ S_R \quad \text{and} \quad A_D:=M \circ S_D.
\end{equation}

\begin{lemma}[Well posedness of $A$]
    The operator $A : \Hd \rightarrow \lzd$ that maps the force $\mathbf{f}$ to the velocity component $v$ of the Stokes problem defined in \eqref{eq:A_def} is well-posed.
\end{lemma}
\begin{proof}
   The well-posedness of $\mathcal{S_R}$ and $\mathcal{S}_D$ are shown in Theorem \ref{lemma:exuni_Stilde} and \ref{th:S_D_well_posed} respectively. 
   Since the measurement operator $M$ and the embedding operator $E$ are linear, bounded, $A$ as the composition of these is well-posed.
\end{proof}

\subsection{Differentiability of the forward maps}

We now derive the analytical tools that allow us to solve the inverse source problem in a stable manner, using a gradient-based regularization methods such as, in our case, the Landweber iteration. 
The core of these  regularized reconstruction methods is the derivation of the derivative and its adjoint. 
This section focuses on derivatives of the embedded Stokes operators $S_R$ and $S_D$ following by their adjoint derivation in Section \ref{sec:adjoints}.
We note that this process was discussed for general elliptic parameter identification problems in \cite{Hoffmann2022} that can also be applied to the Stokes equation. 

\begin{lemma}[Fr\'echet derivative of $S_R$]
 The mapping $S_R$ is Fr\'echet differentiable. The 
 derivative ${S_R}'
 $ in $\mathbf{f}$ is given by
 \begin{displaymath}
  S_R'(\mathbf{f}) : \mathcal{D}(S_R) \to \left(\lzd \times \lzo\right), \quad \mathbf{h} \mapsto (\mathbf{\mathbf{w}},w_p) ,
 \end{displaymath}
 where $\left(\mathbf{\mathbf{w}},w_p\right)$ is the unique weak solution of the boundary value problem
 \begin{equation} \label{eq:stokes_problem_lin}
 \begin{split}
  -\nabla \cdot \big( \eta D(\mathbf{w}) \big) + \nabla w_p &= \mathbf{h} \quad \text{in } \Omega\\
  \nabla \cdot \mathbf{w} &= 0 \quad \text{in } \Omega\\
  \mathbf{w}_{\bot} &= \mathbf{0} \quad \text{on } \partial\Omega \\
  \big(\eta (\nabla \mathbf{w})\mathbf{n} \big)_{\parallel} &= -\lambda \mathbf{w}_{\parallel}\quad \text{on } \partial\Omega.
 \end{split}
\end{equation}
\end{lemma}
\begin{proof}
Since $S_R$ is linear and bounded, the linearzied Stokes equation \eqref{eq:stokes_problem_lin} takes the same form as the original problem; also the G\^ateaux and the Fr\'echet derivatives coincide.

\end{proof}

While the derivative $S_R'$ is identical to $S_R$ for the Robin problem, $S_D$ in the Dirichlet problem is affine-linear. We discuss it in the following. 

\begin{lemma}[Fr\'echet derivative of $S_D$]
 The mapping $S_D$ is Fr\'echet differentiable. The Fr\'echet derivative $S_D'$ in $\mathbf{f}$ is given by
 \begin{equation}
    \label{eq:S_D_frechet}
  S_D'(\mathbf{f}) : \mathcal{D}(S_D) \to \left(\lzd \times \lzo\right), \quad \mathbf{h} \mapsto (\mathbf{w},w_p),
 \end{equation}
 where $\left(\mathbf{\mathbf{w}},w_p\right)$ is the unique weak solution of the boundary value problem
 \begin{displaymath} 
 \begin{split}
  -\nabla \cdot \big( \eta D(\mathbf{w}) \big) + \nabla w_p &= \mathbf{h} \quad \text{in } \Omega\\
  \nabla \cdot \mathbf{w} &= 0 \quad \text{in } \Omega\\
  \mathbf{w} &= \mathbf{0} \quad \text{on } \partial\Omega .
 \end{split}
\end{displaymath}
\end{lemma}

\begin{proof}
The map $S_D$ is affine-linear and bounded. The quantity $(\boldsymbol{\varphi},\varphi_p):=S_D(\mathbf{f}+\mathbf{h}) - S_D(\mathbf{f}) - S_D'(\mathbf{f})\mathbf{h}$ with $S_D'(\mathbf{f})$ as defined  in \eqref{eq:S_D_frechet} solves
 \begin{equation} 
 \begin{split} \label{eq:stokes_trivial}
  -\nabla \cdot \big( \eta D(\boldsymbol{\varphi}) \big) + \nabla \varphi_p &= \mathbf{0} \quad \text{in } \Omega\\
  \nabla \cdot \boldsymbol{\varphi} &= 0 \quad \text{in } \Omega \\
  \boldsymbol{\varphi} &= \mathbf{0} \quad \text{on } \partial\Omega.
 \end{split}
\end{equation}
Since \eqref{eq:stokes_trivial} has a trivial solution, $S_D$ is Fr\'echet differentiable with derivative as in \eqref{eq:S_D_frechet}. 
\end{proof}

\subsection{Adjoint operators\label{sec:adjoints}}
Now in the last step, we construct the adjoint operator for the forward maps $A_R$, $A_D$ via that of $S_R$, $S_D$ and the measurement operator $M$.
We begin with the Robin problem, in which the Fr\'echet derivative $S_R'(\mathbf{f})$ is identical to $S_R$. 

\begin{lemma}[Banach space adjoint of $S_R$]
\label{lemma:banach_robin}
    The Banach adjoint operator of $S_R$ is given by
    \begin{displaymath}
        S_R^{\#}:\left(\lzd \times \lzo\right) \rightarrow \lzd, \quad (\mathbf{w},w_p) \mapsto \boldsymbol{\varphi},
    \end{displaymath} where $\left(\boldsymbol{\varphi},\varphi_p\right)\in \left(\lzd \times \lzo\right)$ solves the adjoint equation
      \begin{equation} \label{eq-adjoint}
 \begin{split}
  -\nabla \cdot \big( \eta D(\boldsymbol{\varphi}) \big) - \nabla \varphi_p &= \mathbf{w} \quad \text{in } \Omega\\
  \nabla \cdot \boldsymbol{\varphi} &= w_p \quad \text{in } \Omega\\
  \boldsymbol{\varphi}_{\bot} &= \mathbf{0} \quad \text{on } \partial\Omega \\
  \big(\eta (\nabla \boldsymbol{\varphi})\mathbf{n} \big)_{\parallel} &= -\lambda \boldsymbol{\varphi}_{\parallel}\quad \text{on } \partial\Omega.
  \end{split}
\end{equation}
\end{lemma}
\begin{proof}  Take any $(\mathbf{w},w_p)\in \lzd \times \lzo$ and any $\mathbf{h}\in\lzd$. 
    Let $(\mathbf{v},p):=S_R(\mathbf{h})$ and $\left(\boldsymbol{\varphi},\varphi_p\right)$ as in \eqref{eq-adjoint}.
    Our aim is to show that the following identity holds
    \begin{equation}
    \langle (\mathbf{w},w_p),S_R(\mathbf{h})\rangle = 
        \langle \mathbf{w}, \mathbf{v} \rangle_{\lzd} + \langle w_p , p \rangle_{\lzo} = \langle \boldsymbol{\varphi} , \mathbf{h} \rangle =
        \langle S_R^{\#} (\mathbf{w},w_p), \mathbf{h} \rangle_{\lzd}.
        \label{eq:adj_cond}
    \end{equation}
    We proceed by rewriting the right-hand side of \eqref{eq:adj_cond}. Firstly, integration-by-part yields
    \begin{equation}
    \begin{split}
         \langle \boldsymbol{\varphi}, - \nabla \cdot (\eta D(\mathbf{v})) \rangle &= \int_\Omega \eta \nabla \boldsymbol{\varphi} : D(\mathbf{v}) \dd \mathbf{x} - \int_{\partial \Omega} \boldsymbol{\varphi}^{T} (\eta D(\mathbf{v}) \mathbf{n} ) \dd \boldsymbol{\sigma} \\
         & = \int_\Omega \eta D(\boldsymbol{\varphi}) : D(\mathbf{v}) \dd \mathbf{x} + \int_{\partial\Omega} \lambda \boldsymbol{\varphi}_{||}^T \mathbf{v}_{||} \dd \boldsymbol{\sigma} \\
         & = \langle \mathbf{v}, - \nabla \cdot (\eta D(\boldsymbol{\varphi}) )\rangle,
         \label{eq:srad_proof_v}
    \end{split}
    \end{equation}
    where the second equality is obtained with symmetries of $D(\mathbf{v})$ and $D(\boldsymbol{\varphi})$.
    Secondly, using the divergence theorem we have 
    \begin{equation}\label{eq:srad_proof_p}
        \begin{split}
            \langle \boldsymbol{\varphi}, \nabla p\rangle = - \int_\Omega  p (\nabla \cdot\boldsymbol{\varphi}) \dd \mathbf{x} + \int_{\partial\Omega} \boldsymbol{\varphi}^T (p \mathbf{n}) \dd \boldsymbol{\sigma} 
                                = \langle - \nabla \cdot \boldsymbol{\varphi}, p\rangle,
        \end{split}
    \end{equation}
    where $\int_{\partial\Omega} \boldsymbol{\varphi}^T (p \mathbf{n}) \dd \boldsymbol{\sigma} = 0$ as $\boldsymbol{\varphi}_\bot = \mathbf{0}$.
    The incompressibility of the velocity $\mathbf{v}$ allows us to introduce $\varphi_p\in\lzo$ and together with the divergence theorem gives
    \begin{equation}\label{eq:srad_proof_incomp}
     0=(\nabla \cdot \mathbf{v} , \varphi_p)= -(\mathbf{v}, \nabla \varphi_p) .
    \end{equation}
    Summing up \eqref{eq:srad_proof_v}, \eqref{eq:srad_proof_p} and \eqref{eq:srad_proof_incomp} leads to
    \begin{equation}
        \begin{split}
            \langle \boldsymbol{\varphi},\mathbf{h}\rangle 
            &= \langle \boldsymbol{\varphi}, - \nabla \cdot (\eta D(\mathbf{v})) - \nabla p \rangle 
            = \langle \mathbf{v}, - \nabla \cdot (\eta D(\boldsymbol{\varphi}) ) - \nabla \boldsymbol{\varphi} \rangle + \langle \nabla \cdot \boldsymbol{\varphi} , p \rangle \\
            & = \langle \mathbf{v},\mathbf{w}\rangle + \langle p, w_p \rangle.
        \end{split}
    \end{equation}
which is the left-hand side of \eqref{eq:adj_cond}, and the proof is complete.
\end{proof}

For the version of the Stokes equation with Dirichlet boundary conditions $S_D'(\mathbf{f})$, the calculation of the Banach adjoint is analogous to Lemma \ref{lemma:banach_robin} and we omit the proof. 

\begin{lemma}[Banach adjoint of $S_D'(\mathbf{f})$]
    \label{lemma:SD_Banach}
    The Banach adjoint of $S'_D(\mathbf{f})$ is given by
    \begin{displaymath}
        {S'_D(\mathbf{f})}^{\#}: \left(\lzd \times \lzo\right) \rightarrow \lzd, \quad (\mathbf{w},w_p) \mapsto \boldsymbol{\varphi},
    \end{displaymath}
    where $(\boldsymbol{\varphi},\varphi_p)\in \left(\lzd \times \lzo\right)$ solves the adjoint equation
        \begin{displaymath}
    \begin{split}
    -\nabla \cdot \big( \eta D(\boldsymbol{\varphi}) \big) - \nabla \varphi_p &= \mathbf{w} \quad \text{in } \Omega\\
    - \nabla \cdot \boldsymbol{\varphi} &= w_p \quad \text{in } \Omega\\
    \boldsymbol{\varphi} &= \mathbf{0} \quad \text{on } \partial\Omega.
    \end{split}
    \end{displaymath}
\end{lemma}

We can now relate the Banach adjoints to the Hilbert space adjoints via the Riesz isomorphism:

\begin{lemma}[Hilbert space adjoint of $S{'}(\mathbf{f})$]
  The Hilbert space adjoint $S{'}(\mathbf{f})^*$ can be computed from the Banach adjoint  $S'(\mathbf{f})^{\#}$ in Lemmas \ref{lemma:banach_robin} and \ref{lemma:SD_Banach} via
  \begin{displaymath}
      S'(\mathbf{f})^{*}:(\lzd \times \lzo) \rightarrow (\Hd)^*,\quad S'(\mathbf{f})^{*}
            :=  I \circ S'(\mathbf{f})^{\#} 
  \end{displaymath}  
    with the Riesz isomorphism 
    \begin{equation}\label{Riesz}
      I:(\Hd)^* \rightarrow \Hd, \quad I\boldsymbol{\varphi}:=(-\Delta_N +\text{Id})^{-1}\boldsymbol{\varphi} =\mathbf{g} ,
    \end{equation}
    meaning that $\mathbf{g}$ solves
    \begin{equation} \label{eq:Riesz}
    \begin{split} 
       - \Delta \mathbf{g} + \mathbf{g} &= \boldsymbol{\varphi} \ \text{in} \ \Omega\\
    (\nabla \mathbf{g})\cdot\mathbf{n} &= \mathbf{0} \ \text{on} \ \partial \Omega .
    \end{split}
    \end{equation}
\end{lemma}
\begin{proof}
    Similarity to \cite{sarnighausen2025_2}, we construct the Hilbert space adjoint using the Gelfand triple $\Hd\subset L^2(\Omega)^d\subset (\Hd)^*$ and the Riesz Representation theorem. More precisely, one has
    \begin{align*}
    (S'(\mathbf{f})^{\#} (\mathbf{v},p) ,\mathbf{f} )_\lzd
    &= \langle S'(\mathbf{f})^{\#} (\mathbf{v},p) ,\mathbf{f} \rangle_{(\Hd)^*\times\Hd}\\
    &= (I_\Hd S'(\mathbf{f})^{\#}(\mathbf{v},p),\mathbf{f} )_{\Hd} =: (S'(\mathbf{f})^* (\mathbf{v},p), \mathbf{f})_\Hd 
    \end{align*}
    with the Riesz isomorphism as in \eqref{Riesz}-\eqref{eq:Riesz}.
\end{proof}

In summary, we explicitly obtain the following adjoints:

\begin{proposition}[Adjoint of $A_R$]\label{prop:adj-Robin}
    The Hilbert space adjoint operator of $A_R$ in \eqref{A} is given by
    \begin{displaymath}
        A_R^*:\lzd \rightarrow H^{-1}(\Omega)^d, \quad \mathbf{w} \mapsto \mathbf{g} = I(\boldsymbol{\varphi}),
    \end{displaymath} with the isomorphism $I$ as in \eqref{Riesz}, where $(\boldsymbol{\varphi},\varphi_p) \in (\lzd\times\lzo)$ solves 
          \begin{equation}\label{eq-adj-zerop} 
 \begin{split}
  -\nabla \cdot \big( \eta D(\boldsymbol{\varphi}) \big) + \nabla \varphi_p &= \mathbf{w} \quad \text{in } \Omega\\
  \nabla \cdot \boldsymbol{\varphi} &= 0 \quad \text{in } \Omega\\
  \boldsymbol{\varphi}_{\bot} &= \mathbf{0} \quad \text{on} \ \partial\Omega \\
  \big(\eta (\nabla \boldsymbol{\varphi})\mathbf{n} \big)_{\parallel} &= -\lambda \boldsymbol{\varphi}_{\parallel}\quad \text{on} \ \partial\Omega.
 \end{split}
\end{equation}

\end{proposition}

\begin{proof}
For all $(\mathbf{w},q)\in (\lzd \times \lzo ) $ it holds that
\begin{displaymath}
        (M(\mathbf{w},q),\mathbf{w}) = (\mathbf{w},\mathbf{w})  
    = (\mathbf{w},\mathbf{w}) + (q,0)  
    = ((\mathbf{w},q),M^{*}(\mathbf{w}))_{\lzd\times\lzo}
    \end{displaymath}
thus the adjoint of the restricted measurement of $M$ is the zero extension operator
    \begin{displaymath}
    M^* : \lzd \rightarrow (\lzd \times \lzo ) , \quad \mathbf{w} \mapsto (\mathbf{w},0).
\end{displaymath}
The adjoint for the operator $A$ is the composition of the adjoint of the embedded Stokes operator $S'(\mathbf{f})^*$ and the adjoint measurment $M^*$. More precisely, we obtain the Hilbert space adjoint of $A$ 
\begin{displaymath}
    A^*: \lzd \rightarrow \Hd  , \quad \mathbf{w} \mapsto (S'(\mathbf{f})^* \circ M^*) (\mathbf{w}) =  S'(\mathbf{f})^*(\mathbf{w},0).
\end{displaymath}
with $(\mathbf{w},0)$ entering the right hand side of the adjoint equation \eqref{eq-adj-zerop}.
\end{proof}

\begin{proposition}[Adjoint of $A_D'(\mathbf{f})$]
    The Hilbert space adjoint operator of $A_D'(\mathbf{f})$ with $A_D$ as in \eqref{A} is given by    \begin{displaymath}
        A'_D(\mathbf{f})^*:\lzd \rightarrow H^{-1}(\Omega)^d, \quad \mathbf{w} \mapsto \mathbf{g} = I(\boldsymbol{\varphi}),
    \end{displaymath}
    with the isomorphism $I$ as in \eqref{Riesz},  where $(\boldsymbol{\varphi},\varphi_p) \in (\lzd\times\lzo)$ solves
          \begin{displaymath} 
 \begin{split}
  -\nabla \cdot \big( \eta D(\boldsymbol{\varphi}) \big) + \nabla \varphi_p &= \mathbf{w} \quad \text{in } \Omega\\
  \nabla \cdot \boldsymbol{\varphi} &= 0 \quad \text{in } \Omega\\
  \boldsymbol{\varphi} &= \mathbf{0} \quad \text{on} \ \partial \Omega .
 \end{split}
\end{displaymath}
\end{proposition}

\begin{proof}
Analogous to the proof of Proposition \ref{prop:adj-Robin}.
\end{proof}

\section{Numerical results} \label{Sec:Numerics}
Finally we apply our results to evaluate synthetic data from simulations as well as measured data from experiments for both actomyosin systems. \\
The simulations were conducted using the Finite Element library \textit{NGSolve} \cite{ngsolve2, ngsolve1}. The discretization method used is based on a H(div)-conforming finite element space and a Hybrid Discontinuous Galerkin (HDG) formulation of the viscous forces using a method proposed in \cite{Lehrenfeld2018} and \cite{Leh16}.
The force was reconstructed from the velocity data using the Landweber iteration method introduced by Landweber in \cite{Landweber1951} for the linear and Hanke, Neubauer and Scherzer in \cite{Han95} for the non-linear case, which in our setting reads as
\begin{equation}
\label{reduced}
\mathbf{f}_{k+1} = \mathbf{f}_k - s A^{*}\left(A(\mathbf{f}_k)-\mathbf{v}^{\operatorname{meas}}\right), \ \ k = 1, 2, 3,\dots
\end{equation}
with constant step size $s>0$.
This is the regularized reconstruction based on the reduced formulation of the forward operator; alternatively one can consider an all-at-one setting \cite{Kaltenbacher:17, Tram19} or a bi-level framework \cite{Tram24, Tram25}, and apply faster regularization techniques such as, e.g., 
the Nesterov-Landweber method \cite{Neubauer17}, or sequential subspace optimization \cite{bhw20,ws16}. We will denote the reconstructed force as $\mathbf{f}^*$.

To ensure stability under noise known as regularization guarantee, the iteration was terminated a posteriori using the discrepancy principle as in \cite{Han95}. More precisely, we stop the iteration at the first time that the residual $r_k:=\norm{A(\mathbf{f}_k)-\mathbf{v}^{\text{meas}}}_\lzd$ fulfills $r_k \leq \tau \delta$ with tolerance $\tau>1$ and noise level $\delta>0$. 
Regarding an initial guess, we use the direct evaluation of the velocity term in the Stokes equation with
\begin{equation}
    \mathbf{f}_0:=- \nabla \cdot \eta D(\mathbf{v}^{\text{meas}}).
    \label{eq:f0}
\end{equation}
We directly evaluate the velocity term of the Stokes equation and $\mathbf{f}_0$ therefore acts as a naive reconstruction.
As described in section \ref{sec:data_aq}, the measured velocity  $\mathbf{v}^{\operatorname{meas}}$ is obtained from the raw video data using PIV. 
Here, the flow is obtained from averaging over the length of the microscopy image series instead of only comparing a small subset of subsequent frames. While this leads to smoother input data $\mathbf{v}^{\operatorname{meas}}$, the naive reconstruction $\mathbf{f}_0$ can still be relatively poor depending on the numerical stability of the used differentiation algorithm as we will see in the next section.

The constant step-size $s$ was determined prior to the iteration using backtracking line search (bls).
The full algorithm is shown in detail in Algorithm \ref{alg:landweber}.

\begin{algorithm}[H]
\caption{Landweber iteration}\label{alg:landweber}
\begin{algorithmic}
\Require Noise level $\delta$, tolerance $\tau$, input data $\mathbf{v}^{\operatorname{meas}}$
\State $\mathbf{f}_0 = - \nabla \cdot \eta D(\mathbf{v}^{\text{meas}})$
\State $s=\operatorname{bls}(\mathbf{f}_0,\mathbf{v}^\text{meas})$
\State $r_{0} = \|A\mathbf{f}_0-\mathbf{v}^\text{meas}\|_{\lzd}$
\State $\mathbf{f} \gets \mathbf{f}_0$
\While{$r_{k}> \tau\delta $ }
    \State $\mathbf{v} = A(\mathbf{f})$ \Comment{$A$ as in \eqref{eq:A_def}}
    \State $r_{k-1}\gets r_k$
    \State $r_k \gets \norm{\mathbf{v}-\mathbf{v}^{\text{meas}}}_\lzd$
    \State $\boldsymbol{\varphi} = S'(\mathbf{f})^{\#}\left(\mathbf{v} - \mathbf{v}^{\text{meas}},0\right)$ \Comment{$S'(\mathbf{f})^{\#}$ as in Lemmas \ref{lemma:banach_robin} and \ref{lemma:SD_Banach}} 
    \State $\mathbf{g} = I(\boldsymbol{\varphi})$ \Comment{$I$ as in \eqref{eq:Riesz}}
    \State $\mathbf{f} \gets \mathbf{f} - s \mathbf{g}$
\EndWhile
\State \Return $\mathbf{f}^* \gets \mathbf{f}$ 
\end{algorithmic}

\end{algorithm}


To evaluate the numerical results, we define the relative noise level $\tilde{\delta}$ and relative reconstruction error $\tilde{\varepsilon}$ as
\begin{displaymath}
    \tilde{\delta} :=\frac{\norm{\mathbf{v}^{\operatorname{meas}}-A(\mathbf{f}^\dagger)}_\lzd }{\norm{A(\mathbf{f}^\dagger)}_{\lzd}}, \ 
    \tilde{\varepsilon} :=\frac{\norm{\mathbf{f}^*-\mathbf{f}^{\dagger}}_\lzd }{\norm{\mathbf{f}^\dagger}_{\lzd}} .
\end{displaymath}
with ground truth force $\mathbf{f}^\dagger$.

\subsection{Regularized reconstruction of forces from synthetic data\label{sec:num_syntetic}}

We first show some examples of reconstructions of forces from simulated velocity fields.
As explained in section \ref{sec:helmholtz}, the volume that the network takes up is insignificant compared to the volume of the surrounding fluid, and we assume that $\nabla \eta = 0$, i.e., that the viscosity is constant.
For the synthetic data the viscosity $\eta$ as well as the boundary parameter $\lambda$ in the Robin version of the Stokes problem were chosen as one. Divergence-free forces are chosen as the ground truth, to enable full comparability to the reconstructions.


\subsubsection{Regularized vs.~naive reconstruction}

At first we compare the reconstructions of forces from noiseless data to the initial guess $\mathbf{f}_0$, which is calculated using the differentiation tools provided by \textit{NGSolve} according to \eqref{eq:f0}. In the absence of noise, we set the stopping criterion according to the monotonicity of the residual by setting the threshold $\tilde{\tau}$, so that the iteration stops when the following criteria is met
\begin{displaymath}
   \left(\norm{A(\mathbf{f}_{k-1})-\mathbf{v}^{\operatorname{meas}}}_\lzd - \norm{A(\mathbf{f}_{k})-\mathbf{v}^{\operatorname{meas}}} _\lzd\right) < \tilde{\tau} \norm{\mathbf{v}^{\operatorname{meas}}}_\lzd. 
\end{displaymath}
For noisy measurements the iteration was additionally stopped using the monotonicity of the residual using machine precision as the threshold.

Figure \ref{fig:droplet_syntetic} shows solutions of the inverse problem for $A_R$ on a $2$-dimensional disk corresponding to the experimental droplets. The noise-free velocity field $\mathbf{v}^{\operatorname{meas}}$, the ground truth force $\mathbf{f}^{\dagger}$, the naive reconstruction $\mathbf{f}_0$ and the reconstruction of the force are depicted. The relative error is $\tilde{\varepsilon}=0.3$ for $\tilde{\tau}=10^{-5}$.
While the velocity generally follows the force and is therefore very similar in direction, the Landweber iteration also recovers the swirls on the lower left. 

Figure \ref{fig:bulk_syntetic} shows an example of the reconstruction of the force on a rectangular domain using Dirichlet boundary conditions corresponding to solving the inverse problem for $A_D$. Similarly to the Robin boundary version the velocity here follows the direction of the force. However, the reconstruction produces some errors close to the boundary, leading to a relative error $\tilde{\varepsilon} = 0.73$ for $\tilde{\tau}=10^{-4}$.

\begin{figure}[H]
    \begin{subfigure}[t]{0.5\textwidth}
        \centering
        \includegraphics[width=\linewidth]{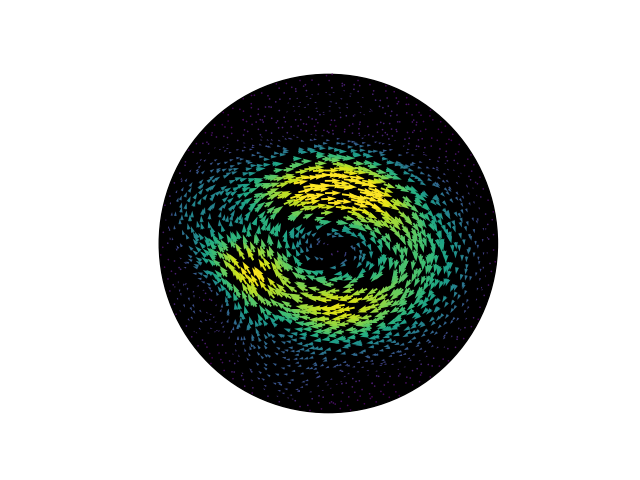}
       \caption{Input data $\mathbf{v}^{\text{meas}}$} 
        \label{fig:droplet_synthetic_v_real}
    \end{subfigure}
    \begin{subfigure}[t]{0.5\textwidth}
        \centering
        \includegraphics[width=\linewidth]{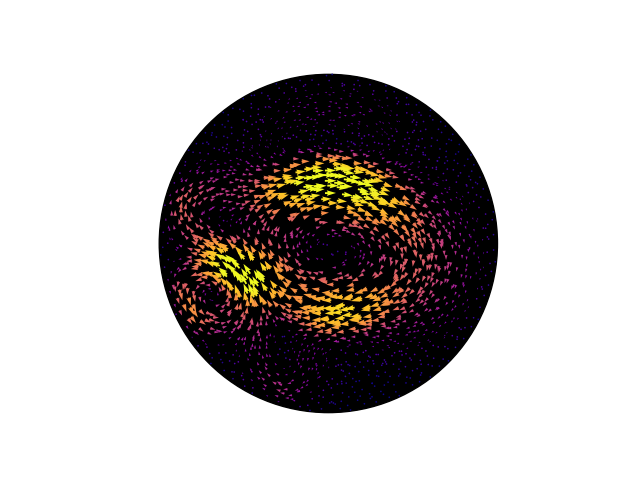}
        \caption{Ground truth force $\mathbf{f}^{\dagger}$}
        \label{fig:droplet_synthetic_f_real}
    \end{subfigure}\\
    \begin{subfigure}[t]{0.5\textwidth}
        \centering
        \includegraphics[width=\linewidth]{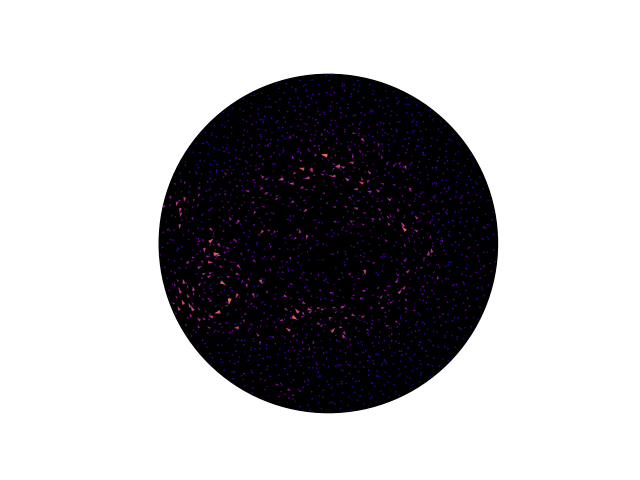}
        \caption{Naive reconstruction $\mathbf{f}_0$}
    \end{subfigure}
    \begin{subfigure}[t]{0.5\linewidth}
        \centering
        \includegraphics[width=\linewidth]{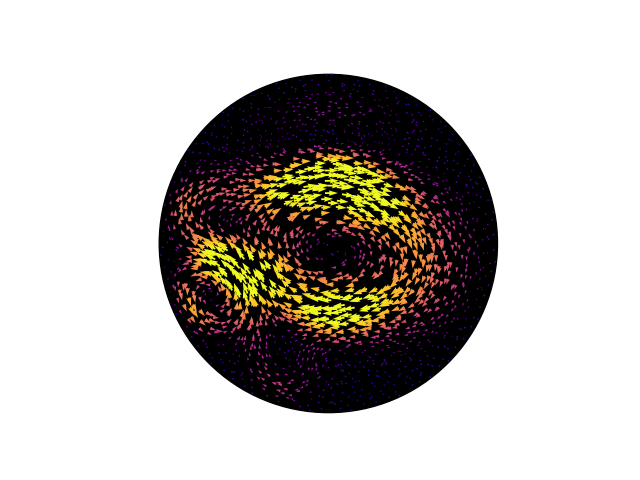}
        \caption{Reconstructed force $\mathbf{f}^*$}
    \end{subfigure}\\
    \begin{subfigure}[t]{0.5\linewidth}
        \centering
        \includegraphics[width=\linewidth]{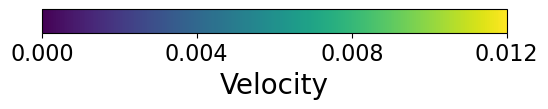}
    \end{subfigure}
    \begin{subfigure}[t]{0.5\linewidth}
        \centering
        \includegraphics[width=\linewidth]{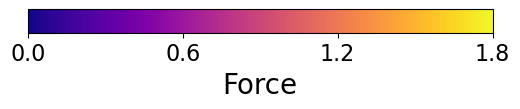}
    \end{subfigure}
    \caption{Velocity field $\mathbf{v}^{\operatorname{meas}}$, ground truth $\mathbf{f}^\dagger$, initial guess $\mathbf{f}_0$ and reconstruction of force $\mathbf{f}^*$ for the operator $A_R$ with Robin boundary conditions.}
    \label{fig:droplet_syntetic}
\end{figure}

\begin{figure}[H]
    \begin{subfigure}[t]{0.5\textwidth}
        \centering
        \includegraphics[width=\linewidth]{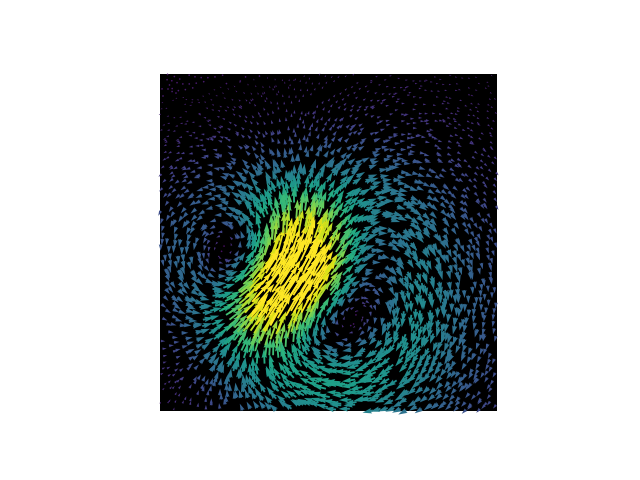}
       \caption{Input data $\mathbf{v}^{\text{meas}}$} 
    \end{subfigure}
    \begin{subfigure}[t]{0.5\textwidth}
        \centering
        \includegraphics[width=\linewidth]{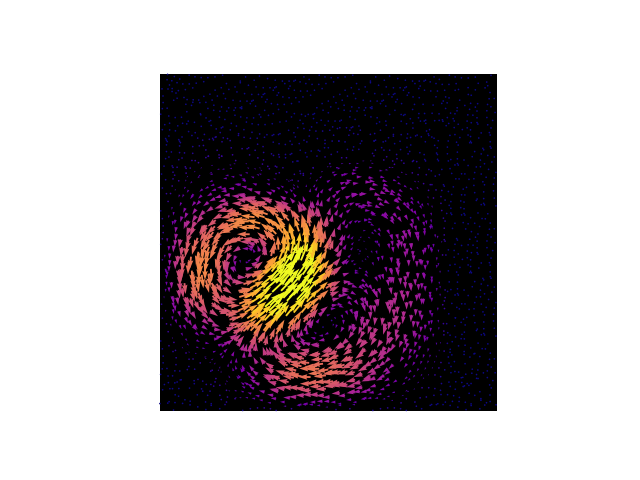}
        \caption{Ground truth $\mathbf{f}^{\dagger}$}
    \end{subfigure}\\
    \begin{subfigure}[t]{0.5\textwidth}
        \centering
        \includegraphics[width=\linewidth]{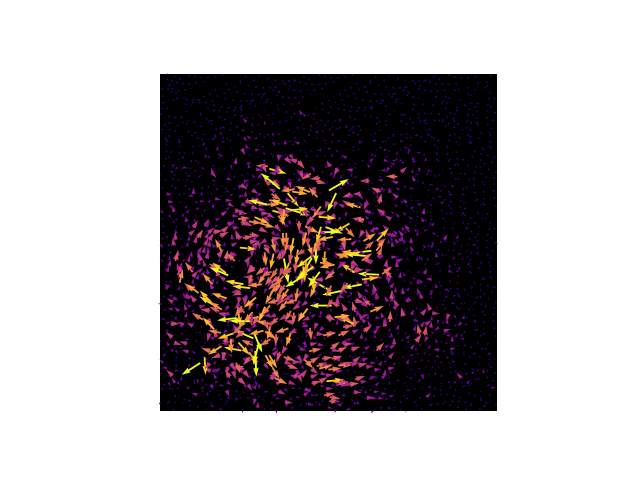}
        \caption{Naive reconstruction $\mathbf{f}_0$}
    \end{subfigure}
    \begin{subfigure}[t]{0.5\linewidth}
        \centering
        \includegraphics[width=\linewidth]{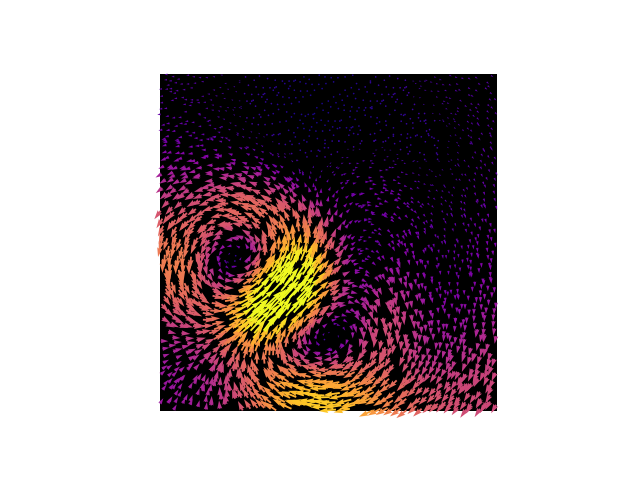}
        \caption{Reconstructed force $\mathbf{f}^*$}
    \end{subfigure}
    \begin{subfigure}[t]{0.5\linewidth}
        \includegraphics[width=\linewidth]{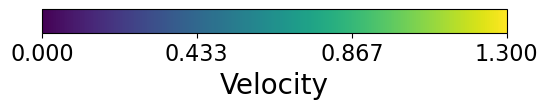}
    \end{subfigure}
    \begin{subfigure}[t]{0.5\linewidth}
        \centering
    \includegraphics[width=\linewidth]{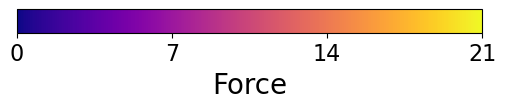}
    \end{subfigure}
    \caption{Velocity field $\mathbf{v}^{\operatorname{meas}}$, ground truth $\mathbf{f}^\dagger$, initial guess $\mathbf{f}_0$ and reconstructed force $\mathbf{f}^*$ for the operator $A_D$ with Dirichlet boundary conditions. }
    \label{fig:bulk_syntetic}
\end{figure}

\subsubsection{Reconstructions of forces from noisy synthetic data}
The reconstruction of the force was also tested on velocity data of this dataset with uniform random noise for the Robin version of the problem with $\tau=1.01$. Three examples of noisy velocity fields and corresponding reconstructions are plotted in figure \ref{fig:droplet_noisy}. 
 
\begin{table}[H]
    \begin{tabular}{l|lllllllll}
        Rel. noise level $\tilde{\delta}$ & 0.15&0.3&0.45&0.62&0.76&0.91&1.1&1.2&1.37\\
       Rel. reconstruction error $\tilde{\varepsilon}$ &0.2&0.21&0.23&0.22&0.25&0.27&0.26&0.28&0.29\\
    \end{tabular}
\end{table}

 While the structure and magnitude of the force are retained relatively well even for high noise levels, some details of the force get lost as the noise level increases. 

\begin{figure}[H]
    \begin{subfigure}[t]{0.5\textwidth}
        \centering
        \includegraphics[width=\linewidth]{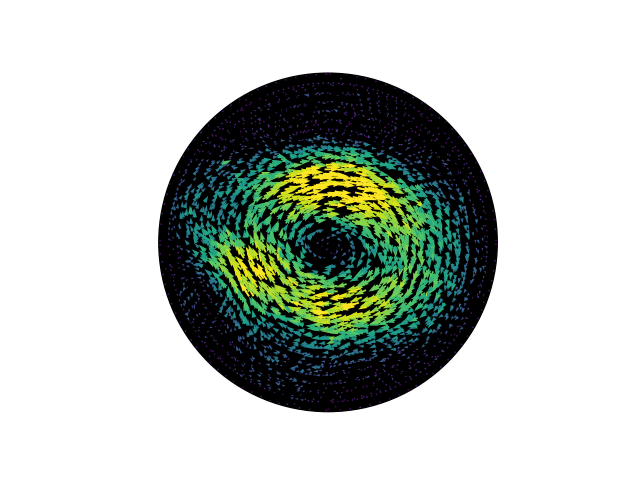}
       \caption{Noisy velocity $\mathbf{v}^{\text{meas}}$ with $\tilde{\delta}=0.15$} 
    \end{subfigure}
    \begin{subfigure}[t]{0.5\textwidth}
        \centering
        \includegraphics[width=\linewidth]{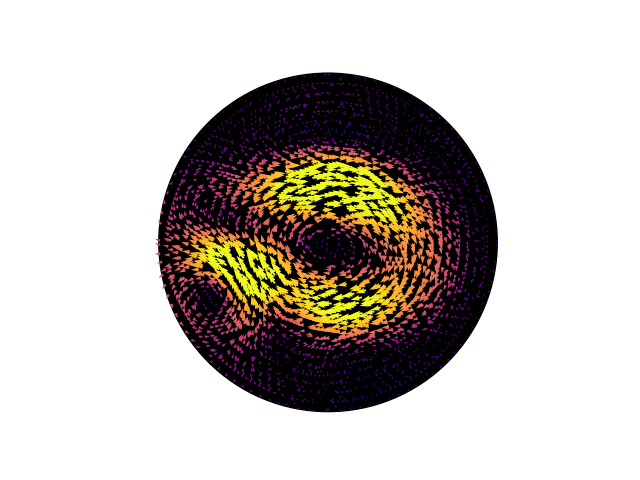}
        \caption{Reconstructed force $\mathbf{f}^*$ with $\tilde{\varepsilon}=0.23$.}
    \end{subfigure}\\
    \begin{subfigure}[t]{0.5\textwidth}
        \centering
        \includegraphics[width=\linewidth]{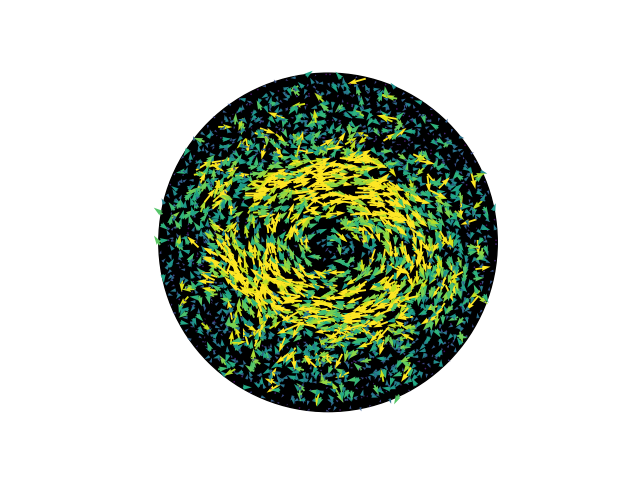}
        \caption{Noisy velocity $\mathbf{v}^{\text{meas}}$ with $\tilde{\delta}=0.76$}
    \end{subfigure}
    \begin{subfigure}[t]{0.5\linewidth}
        \centering
        \includegraphics[width=\linewidth]{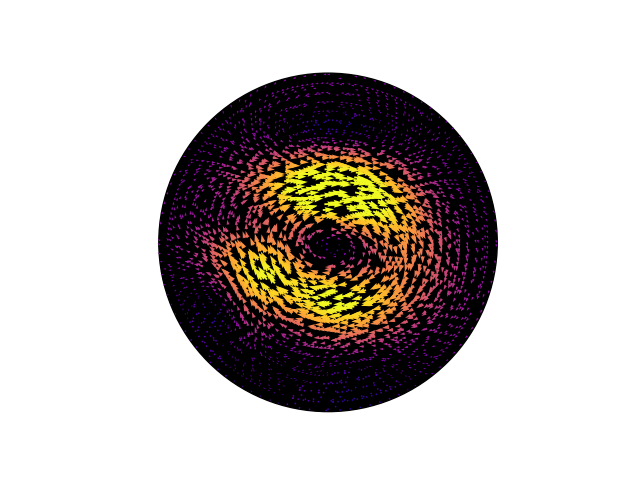}
        \caption{Reconstructed force $\mathbf{f}^*$ with $\tilde{\varepsilon}=0.25$.}
    \end{subfigure}\\
        \begin{subfigure}[t]{0.5\textwidth}
        \centering
        \includegraphics[width=\linewidth]{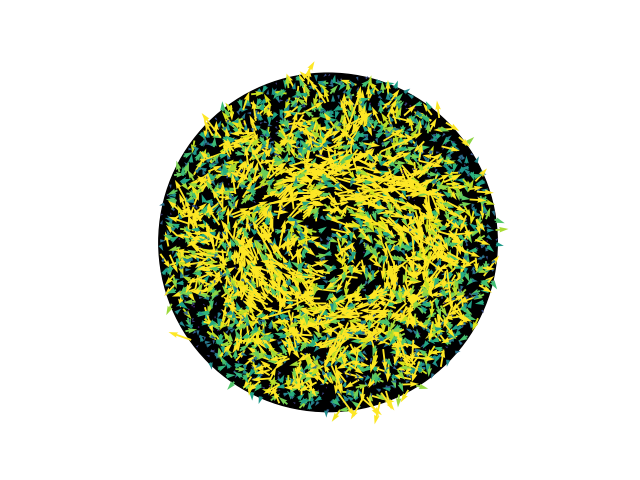}
        \caption{Noisy velocity $\mathbf{v}^{\text{meas}}$ with $\tilde{\delta}=1.37$}
    \end{subfigure}
    \begin{subfigure}[t]{0.5\linewidth}
        \centering
        \includegraphics[width=\linewidth]{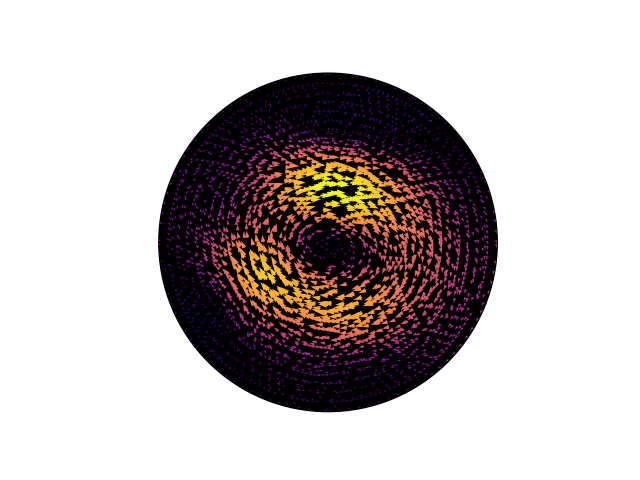}
        \caption{Reconstructed force $\mathbf{f}^*$ with relative reconstruction error $\tilde{\varepsilon}=0.29$.}
    \end{subfigure}\\
    \begin{subfigure}[t]{0.5\linewidth}
        \centering
        \includegraphics[width=\linewidth]{figures/reconstructions/droplet_noisy_1_v_noisy_cbar.png}
    \end{subfigure}
    \begin{subfigure}[t]{0.5\linewidth}
        \centering
        \includegraphics[width=\linewidth]{figures/reconstructions/droplet_noisy_1_f_pred_cbar.png}
    \end{subfigure}
    \caption{Reconstructions of forces (right) from velocity fields with uniform random noise and differing noise levels (left).}
    \label{fig:droplet_noisy}
\end{figure}

\subsection{Reconstructions from experimental data} \label{Sec:ExpData}

The reconstruction algorithm \ref{alg:landweber} was also tested on experimental data for both versions of the Stokes problem, the bulk phase and actomyosin networks confined in a droplet. In this section the vector fields are plotted over an image of the actomyosin network at a midway point of the measured video. 

In the experimental droplets and bulk measurements, the fluid containing the acto-myosin network is almost exclusively deionized water. Additionally, the network only takes up a small part of the absolute volume of the solution as discussed in section \ref{sec:helmholtz}. Therefore, we will also use a constant viscosity $\eta(\mathbf{x}) \equiv 1$ for the experimental data. 


\begin{figure}[H]
    \begin{subfigure}[t]{0.5\textwidth}
        \centering
        \includegraphics[width=\linewidth]{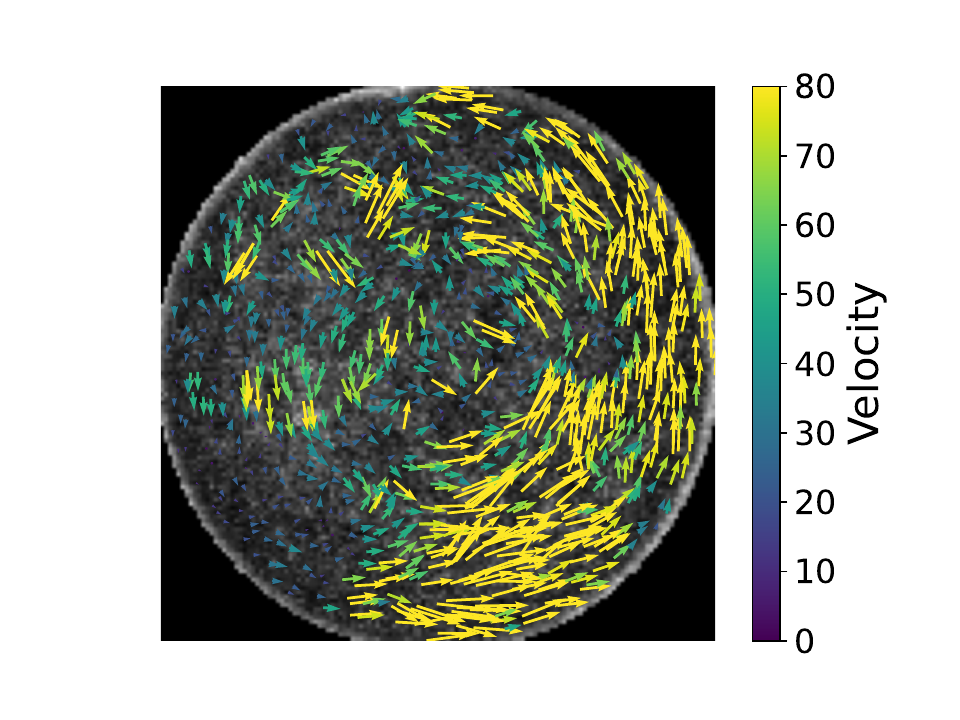}
       \caption{Measured velocity flow field $\mathbf{v}^{\text{meas}}$ in $[\frac{\mu m}{\operatorname{min}}]$} 
        \label{fig:v_real_droplet}
    \end{subfigure} 
    \begin{subfigure}[t]{0.5\textwidth}
        \centering
        \includegraphics[width=\linewidth]{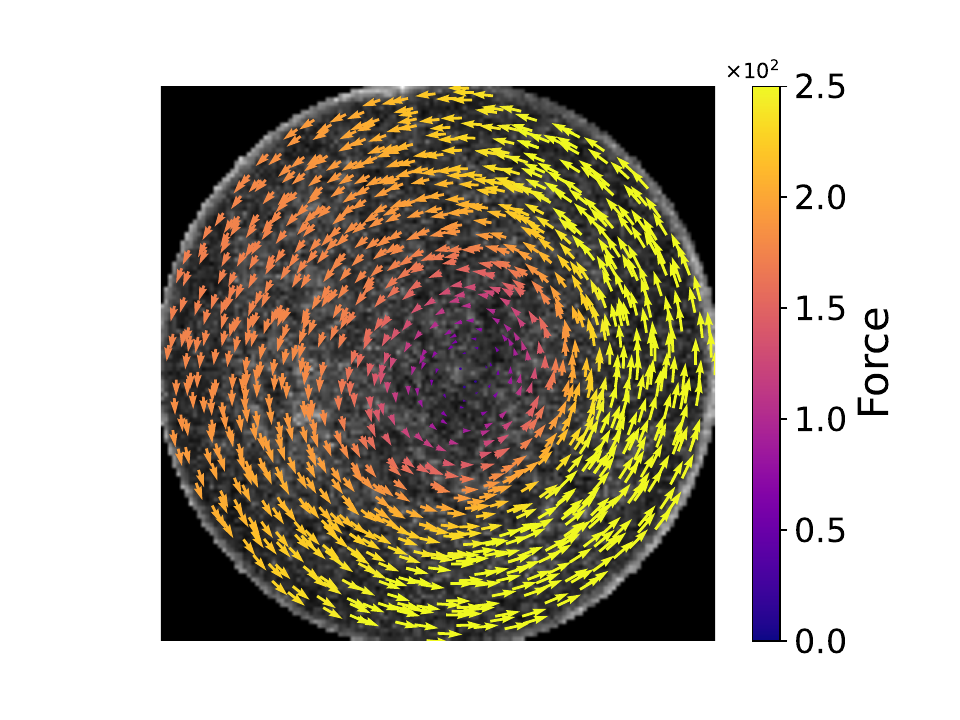}
        \caption{Reconstructed force $\mathbf{f}^*$ in $[\frac{pN}{(\mu m)^3}]$}
        \label{fig:f_real_droplet}
    \end{subfigure}
    \caption{Reconstruction of a force causing the observed flow of an actin-myosin network encapsulated in a droplet, corresponding to the inverse problem for the operator $A_R$.}
    \label{fig:droplet_experimental}
\end{figure}

Figure \ref{fig:droplet_experimental} shows the measured velocity field and the reconstruction of the force exerted by an active acto-myosin network encapsulated in a droplet. Robin boundary conditions were used to represent the oil-water boundary.
The data preparation process is discussed in section \ref{sec:data_aq} and the velocity field obtained by PIV from fluorescence microscopy videos is shown in figure \ref{fig:v_real_droplet}. Since the velocity field $\mathbf{v}^{\operatorname{meas}}$ obtained from the video data is given on a rectangular grid, it is here first interpolated onto a rectangular finite element mesh and then transferred to a circular mesh. This effectively ``cuts off'' the corners, e.g., the area of the measurement field outside the confinements of the droplet.  
The resulting reconstruction of the divergence-free component with relative noise level $\tilde{\delta}=0.43$ and $\tau = 1.7$ of the force is shown in \ref{fig:f_real_droplet}. 
As expected, the flow inside the droplet follows the direction of the divergence-free force.

\begin{figure}[H]
    \begin{subfigure}[t]{0.5\textwidth}
        \centering
        \includegraphics[width=\linewidth]{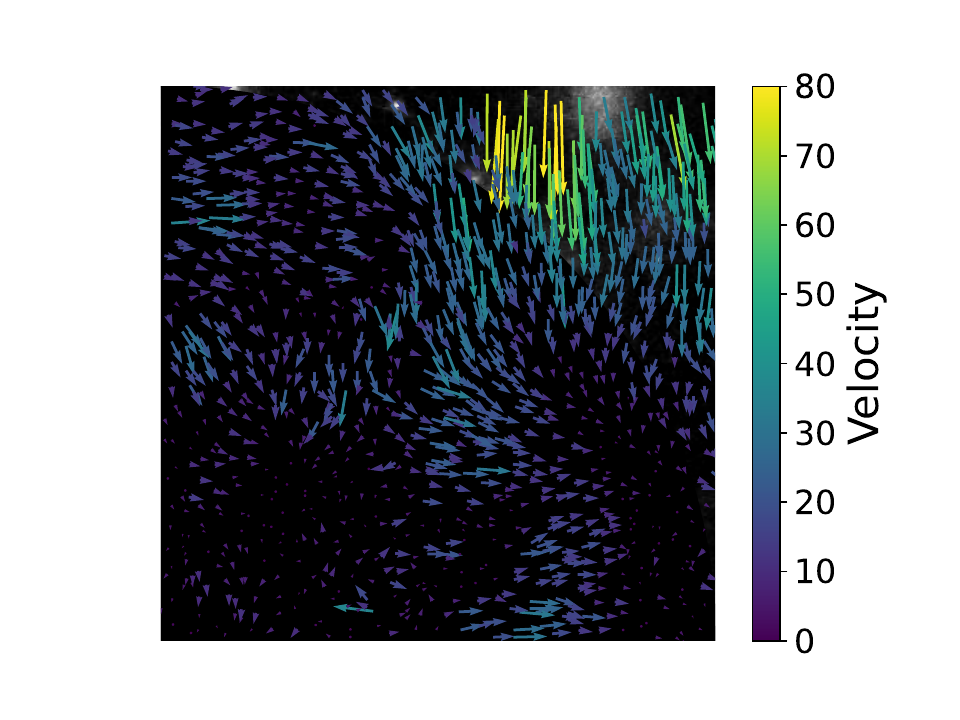}
       \caption{Input data $\mathbf{v}^{\text{meas}}$ in $[\frac{\mu m}{\operatorname{min}}]$} 
        \label{fig:v_bulk_45_real}
    \end{subfigure} 
    \begin{subfigure}[t]{0.5\textwidth}
        \centering
    \includegraphics[width=\linewidth]{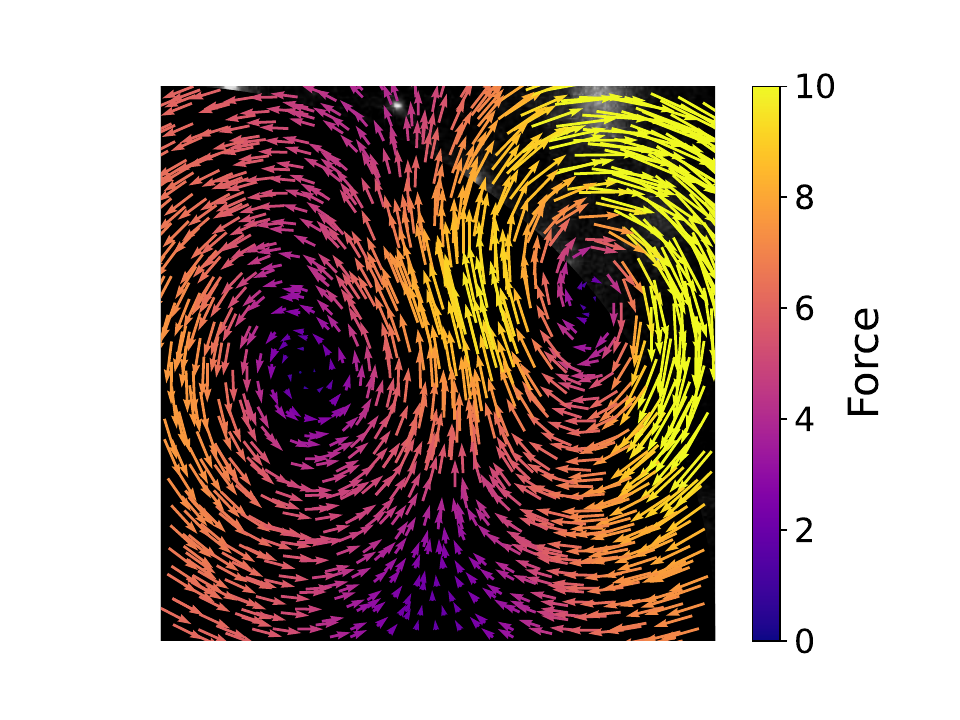}
        \caption{Reconstructed force $\mathbf{f}^*$ in $[\frac{p N}{ (\mu m)^3}]$}
        \label{fig:f_bulk_45_real}
    \end{subfigure}
    \caption{Experimentally obtained velocity field of the flow of a bulk actin-myosin network and reconstructed force $\mathbf{f}^*$ using $A_D$. }
    \label{fig:bulk_experimental_45}
\end{figure}

\begin{figure}[H]
    \begin{subfigure}[t]{0.5\textwidth}
        \centering        \includegraphics[width=\linewidth]{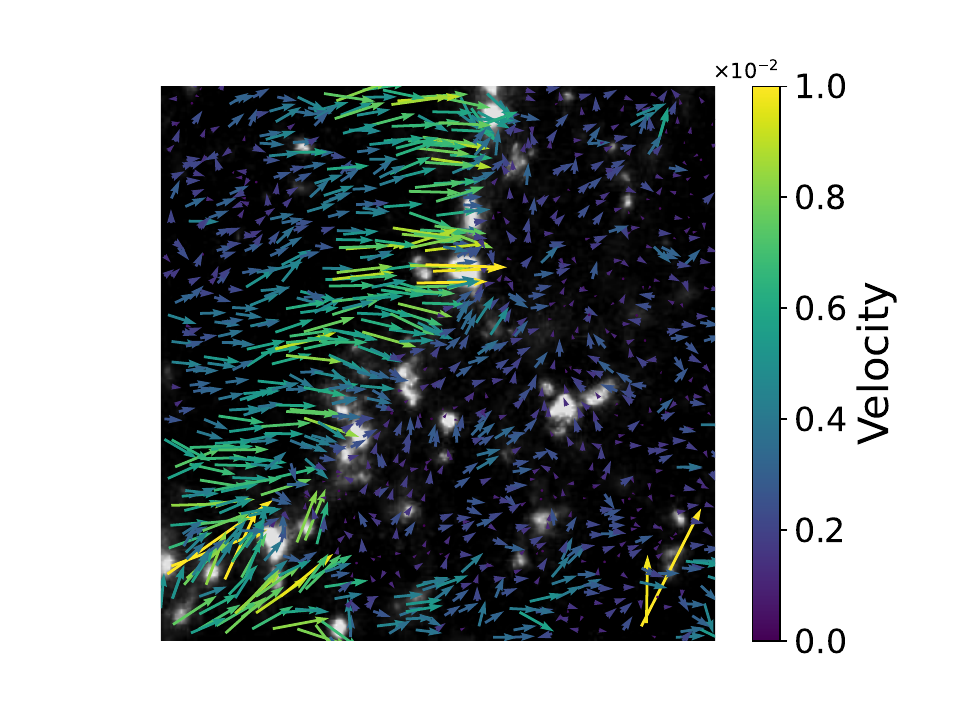}
       \caption{Input data $\mathbf{v}^{\text{meas}}$ in $[\frac{\mu m}{\operatorname{min}}]$} 
        \label{fig:v_bulk_40_real}
    \end{subfigure} 
    \begin{subfigure}[t]{0.5\textwidth}
        \centering
    \includegraphics[width=\linewidth]{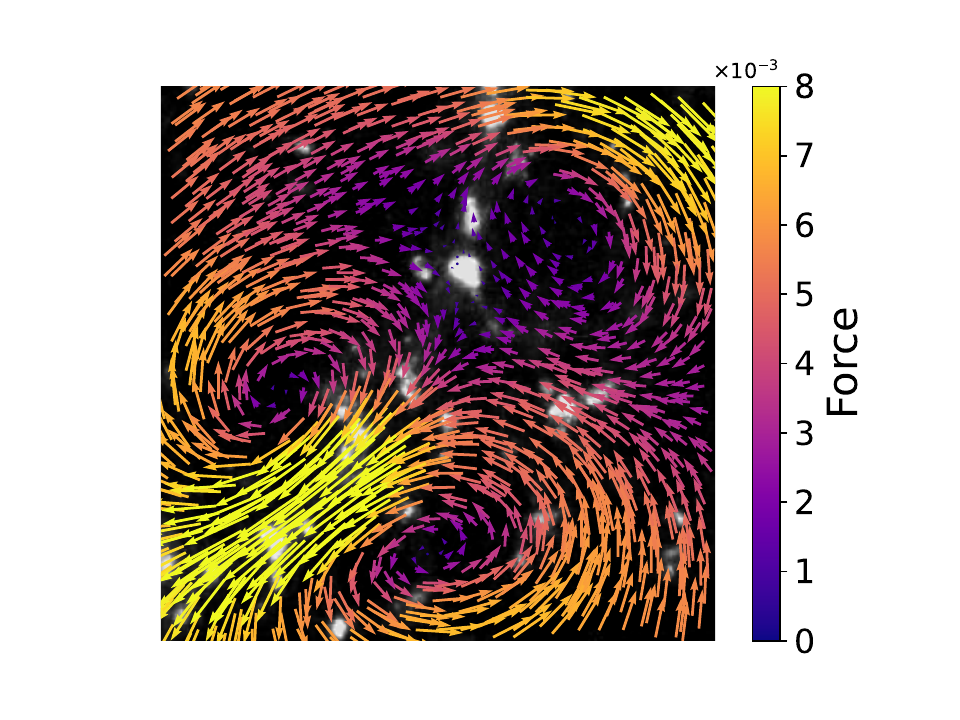}
        \caption{Reconstructed force $\mathbf{f}^*$ in $[\frac{\mu N}{(\mu m)^3}]$}
        \label{fig:f_bulk_40_real}
    \end{subfigure}
    \caption{Experimentally obtained velocity field $\mathbf{v}^{\operatorname{meas}}$ of the flow of a bulk actin-myosin network and reconstructed force $\mathbf{f}^*$ using $A_D$.}
    \label{fig:bulk_experimental_40}
\end{figure}

In Figure \ref{fig:bulk_experimental_45} the reconstruction method was tested on velocity data of a network that was not encapsulated, corresponding to the version of the Stokes problem with Dirichlet boundary conditions. The relative noise level was set to $\tilde{\delta}=0.41$ and the tolerance to $\tau=1.7$.

Figure \ref{fig:bulk_experimental_40} shows the actin velocity flow of a different section of the sample shown in figure \ref{fig:bulk_experimental_45} and its corresponding reconstructed force $\mathbf{f}^*$.
Here the velocity field shows contraction of the actin network, resulting in the sharp decline in movement in the left half of the measured area. For the reconstruction, we set the relative noise level to $\tilde{\delta}=0.49$ and $\tau=1.7$. 

We expect that the missing information on the irrotational component of the active forces is more significant in the bulk system than in the droplet scenario, since an encapsulation typically means that rotations are more prominent. For this reason, the reconstruction is more easily related to the measured velocity field in the droplet case than in the bulk case. In addition, it is not clear that the observed flow originates from the activity of the actomyosin or from an interaction of the fluid with the ambient outside fluid.

\section{Discussion and outlook} \label{Sec:5}

The aim of this work is to provide mathematical tools to identify forces exerted by an active actomyosin network in a fluid. 
We used the Stokes equation together with an incompressibility condition and suitable boundary value conditions to reflect the physical setting to describe the fluid flow, which is driven by the activity of the actomyosin network. The respective force appears as a source term in the Stokes equation. 
To this end, we formulated a mathematical framework based on the assumption that the fluid can be considered Newtonian. The analysis reveals that only the divergence-free part of the force field $\mathbf{f}$ can be reconstructed from the knowledge of $\mathbf{v}$ for a known viscosity, which was also demonstrated in numerical experiments for synthetic data with a known ground truth. Furthermore, we have been able to extend our regularization algorithm to evaluate experimental data. 

Our findings thus complement a class of applied inverse source problems that aim at identifying forces, for example traction force microscopy, see, e.g., \cite{Ambrosi2006, Schwarz2015}, 
where the forces that a cell exerts on an elastic substrate are to be reconstructed to obtain insights into cell mechanics. The underlying mathematical problem is an inverse source problem for the elastic substrate, see \cite{Sarnighausen_2025}. Inverse source problems for force identification also arise in applications besides cell physics, e.g., in defect localization in composites \cite{Binder_2015} or in elastodynamics in seismology \cite{Bao_2018}.

In the future, to analyze actual three-dimensional droplets and to understand the development of active forces when the actomyosin network evolves, we will extend our model to incorporate the inherent time-dependence of actomyosin systems, and apply suitable regularization techniques. In addition, our framework will be adapted to the respective data acquisition process that is typically limited to imaging the system only within the focal plain, leading to even more restrictions in data availability.
The full reconstruction of the active force $\mathbf{f}$ inside of an actual cell requires knowledge of the viscosity $\eta$ and the pressure $p$. In further studies the aim is to substitute these quantities with approximations using experimental data and thus reconstruct the rotation-free part as an additional inverse problem.

\section{Data availability statement}
The codes and data that support the findings of this article are openly available on GRO.data~\cite{5SY2VM_2025} at \url{https://doi.org/10.25625/5SY2VM}.

\section{Acknowledgements}\label{sec:acknowledgments}
A.W., A.J., and S.K, acknowledge funding by Deutsche Forschungsgemeinschaft (DFG, German Research Foundation) – Project-ID 449750155 – RTG 2756, Projects A4, A6, A7.
E.K., A.W.~and T.N.~additionally acknowledge support by DFG CRC 1456 (Project-ID 432680300, Projects B06 and Project C04).

\appendix
\section{Appendix}

\begin{lemma}\label{lemma:div_appendix}
The following identity regarding the divergence of the velocity term in the Stokes equation holds:
    \begin{displaymath}
        \nabla \cdot ( -  \nabla \cdot (\eta D(\mathbf{v}))) 
        = - \left( 2 \nabla \eta \cdot \Delta \mathbf{v} + D(\mathbf{v}) : \nabla(\nabla\eta) \right).
    \end{displaymath}
    As a consequence, the right hand side vanishes for constant viscosity $\eta$.
\end{lemma}
\begin{proof}
Using the product rule, the divergence can be rewritten as
    \begin{equation}\label{eq:div_appendix_1}
    \nabla \cdot (\nabla \cdot (\eta D(\mathbf{v})) )
     = \nabla \cdot ( D(\mathbf{v}) \cdot \nabla \eta   + \eta \Delta \mathbf{v} )
    \end{equation}
The first term of \eqref{eq:div_appendix_1} then becomes
\begin{displaymath}
    \nabla \cdot (D(\mathbf{v}) \cdot \nabla \eta   ) 
= \Delta \mathbf{v} \cdot \nabla \eta  + D(\mathbf{v}) : \nabla (\nabla \eta) 
\end{displaymath}
and the second term is 
\begin{displaymath}
    \nabla \cdot ( \eta \ \Delta \mathbf{v} ) 
= \Delta \mathbf{v} \cdot \nabla \eta + \eta \nabla \cdot ( \Delta \mathbf{v}) 
=  \Delta \mathbf{v} \cdot \nabla \eta ,
\end{displaymath}
where $ \nabla \cdot \Delta \mathbf{v}  = 0$ with Graßmann's identity and the incompressibility of $\mathbf{v}$. 
\end{proof}

\vspace*{2ex}

\bibliographystyle{plain}
\bibliography{bibliography}

\pagestyle{myheadings}
\thispagestyle{plain}

\end{document}